\documentclass[reprint,pra,twocolumn]{revtex4-1}

\usepackage{amsthm} 

\usepackage{amsmath}
\usepackage{latexsym}
\usepackage{amsfonts}
\usepackage{amssymb}
\usepackage{mathtools} 
\usepackage{color}
\usepackage{bbm,dsfont}
\usepackage{graphicx}
\usepackage{ifthen}
\usepackage{enumerate}

\usepackage[T1]{fontenc}
\usepackage[utf8]{inputenc}

\newtheorem{proposition}{Proposition}
\newtheorem{theorem}{Theorem}
\newtheorem{lemma}{Lemma}
\newtheorem{corollary}{Corollary}

\theoremstyle{definition}

\newtheorem{remark}{Remark}

\newtheorem{definition}{Definition}

\newcommand{\R}{\mathbb{R}} 
\newcommand{\C}{\mathbb{C}} 
\newcommand{\F}{\mathbb{F}} 
\newcommand{\Z}{\mathbb Z} 
\newcommand{\T}{\mathbb T} 
\newcommand{\Tr}[1]{{\rm Tr} #1} 

\newcommand{\abs}[1]{\left| #1 \right|} 
\newcommand{\tr}[1]{\mathrm{tr}\left[#1\right]} 
\let\baraccent=\= 
\renewcommand{\=}[1]{\stackrel{#1}{=}} 

\newcommand{\hi}{\mathcal{H}} 
\newcommand{\lh}{\mathcal{L(H)}} 
\newcommand{\elle}[1]{\mathcal{L}(#1)} 
\newcommand{\kb}[2]{|#1\rangle\langle#2|} 
\newcommand{\no}[1]{\left\|#1\right\|} 
\newcommand{\id}{I} 
\newcommand{\fii}{\varphi}

\newcommand{\cp}{\mathcal{P}}
\newcommand{\cpk}{\mathcal{P}^{(k)}}
\newcommand{\cpl}{\mathcal{P}^{(\ell)}}

\newcommand{\sfe}{{\sf E}}
\newcommand{\sfp}{{\sf P}}

\newcommand{\sfg}{{\sf G}}

\newcommand{\Eo}{\mathsf{E}}
\newcommand{\Go}{\mathsf{G}}
\newcommand{\Po}{\mathsf{P}}
\newcommand{\Qo}{\mathsf{Q}}

\newcommand{\vl}{\mathbf{l}} 
\newcommand{\vu}{\mathbf{u}} 
\newcommand{\vv}{\mathbf{v}} 
\newcommand{\vw}{\mathbf{w}} 
\newcommand{\vnull}{\mathbf{0}}
\newcommand{\pair}[2]{\langle #1, #2\rangle}

\newcommand{\sym}[2]{\left[\,#1\, , \,#2\,\right]} 
\newcommand{\dual}[2]{\left\langle\,#1\, , \,#2\,\right\rangle} 

\newcommand{\ff}{\mathcal{F}}

\newcommand{\sas}{\mathrm{S}}

\begin{document}

\title{An operational link between MUBs and SICs}

\author{Roberto Beneduci}
\email{roberto.beneduci@unical.it}
\affiliation{\small Dipartimento di Matematica, Universit\`a della Calabria, Cosenza, 
and INFN gruppo collegato Cosenza, Cosenza, Italy}

\author{Thomas J Bullock}
\email{tjb525@york.ac.uk}
\affiliation{\small Department of Mathematics, University of York, York, UK}

\author{Paul Busch}
\email{paul.busch@york.ac.uk}
\affiliation{\small Department of Mathematics, University of York, York, UK}

\author{Claudio Carmeli}
\email{claudio.carmeli@gmail.com}
\affiliation{\small D.I.M.E., Universit\`a di Genova, Savona, Italy}

\author{Teiko Heinosaari}
\email{teiko.heinosaari@utu.fi}
\affiliation{\small Turku Centre for Quantum Physics, Department of Physics and Astronomy, University of Turku, Finland}

\author{Alessandro Toigo}
\email{alessandro.toigo@polimi.it}
\affiliation{\small Dipartimento di Matematica, Politecnico di Milano, Milano, and I.N.F.N.,
Sezione di Milano, Italy}

\date{\small \today}

\begin{abstract}
We exhibit an operational connection between mutually unbiased bases  and symmetric 
informationally complete positive operator-valued measures.
Assuming that the latter exists, we show that there is a strong link between these two structures in all prime power dimensions.
We also demonstrate that a similar link cannot exist in dimension 6.
\end{abstract}

\maketitle

\section{Introduction}

The existence of {\em symmetric informationally complete positive operator valued 
measures} (SIC POVMs, or simply  SICs) in arbitrary dimensions and the existence of complete 
systems of $d+1$ mutually unbiased bases (MUBs) in complex Hilbert spaces of dimensions 
$d$ other than prime power have remained unsolved despite much effort expended on 
them over the last ten years or so. Connections with finite geometries have been exhibited
for both SICs and MUBs, including the intriguing fact that the relevant geometries for either 
of them are equivalent (where they exist) \cite{Grassl2004,Wootters2006}. The construction 
schemes proposed so far for SICs and MUBs are essentially mathematically motivated and 
based on combinatorial, group theoretic, algebraic or finite geometric structures 
\cite{Zauner1999,Bandyopadhyay2002,Gibbons2004,Renes2004,Saniga2004,Howe2005,Gurevich2008,Appleby2009,
Paterek2009,Durt2010,ScottGrassl2010,Revzen2012}.

Here we elucidate a physically motivated connection between SICs and MUBs, which 
emerges quite naturally in some respects.
Our main tools will be the generalisation of a SIC POVM into a \emph{SIC system} and of MUBs into \emph{mutually unbiased POVMs}.

The basic idea is easily established by 
considering the qubit case: there a SIC is found to be a natural joint observable of three 
POVMs that are smearings of three mutually unbiased binary observables. 
It turns out that these POVMs have a similar relation as mutually unbiased bases and they will therefore be referred to as mutually unbiased POVMs.  
Characteristic properties of mutually unbiased POVMs arising as marginals of SICs in dimension $d$ are established; these properties will be 
referred to as SIC compatibility conditions. 
We show that whenever a complete system of mutually unbiased commutative POVMs is obtained from a SIC, it is possible to extract a complete system of MUBs.

The condition for a SIC possessing $d+1$ mutually unbiased POVMs as marginals is found to be equivalent to 
there being $d-1$ orthogonal Latin squares of size $d$. Conversely, given a complete 
system of mutually unbiased POVMs for dimension $d$, a system of operators can be canonically constructed 
such that it contains these mutually unbiased POVMs as marginals whenever there are $d-1$ orthogonal Latin squares 
of order $d$. This system possesses all properties of a SIC except, possibly, positivity. Such a
collection of operators will be called a SIC system. For each complete MUB system there is an infinite
variety of SIC systems. It is an open problem to identify SICs among them where they exist.

Finite covariant phase space observables constitute a concrete application of the general framework sketched above. Here, we recall that finite covariant phase space observables are POVMs based on finite phase spaces (see \cite{Gibbons2004}) and are covariant with respect to the discrete Weyl-Heisenberg representation acting on a $d=p^n$-dimensional Hilbert space (with $p$ a prime number different from $2$, and $n$ a positive integer). We show that every finite covariant phase space observable admits $d+1$ mutually unbiased and commutative marginal POVMs that are smearings of $d+1$ MUBs. As an application of the general results developed in the first part of the paper, in Theorem \ref{theo:opcharct} we then characterize SIC covariant phase space observables in terms of a particular property of their mutually unbiased marginal POVMs.

Our construction scheme for MUBs from SICs and vice versa may be seen as a possible realisation 
of the construction scheme outlined by Wootters \cite{Wootters2006}; 
the equivalence of the associated finite geometries pointed out there is found here
to systematically underpin the SIC-MUB bridge via the marginality relation employed.

\section{MUBs, SICs, and their connection in $d=2$}

\subsection{MUBs and their generalisation}

Throughout the paper the underlying complex Hilbert space will be of dimension $d$ and denoted $\hi_d$
or occasionally $\mathbb{C}^d$. 
A basic structure to be used is a function $k\mapsto \Eo(k)$ from some finite set $\Omega$ of outcomes to the set of positive operators on $\hi_d$ with the property that $\sum_k \Eo(k)=\id$, the identity operator on $\hi_d$. 
This kind of function determines a \emph{positive operator valued measure}, in short, POVM. 
We refrain from using the appropriate measure
theoretic language and simply refer to $\sfe$ as a POVM. If the elements of a POVM are projections we will
call $\sfe$ a projection valued, or spectral, measure (PVM).

We begin by recalling the basic definition.
\begin{definition}\label{def:mub}
A family of orthonormal bases $\{ \fii^{k}_i \}_{i=0}^{d-1}$, $k=1,2,\dots, n$ of  a 
$d$-dimensional Hilbert space $\mathcal{H}_d$ are said to be \emph{mutually unbiased} if, for any $i,j$, 
\begin{equation}\label{eqn:mu-bases}
\bigl|{\bigl\langle{\fii^{k}_i}\big|{\fii^{\ell}_j}\bigr\rangle}\bigr|^2 = \frac 1d\,(1-\delta_{k\ell})+\delta_{k\ell}\delta_{ij}.
\end{equation}
\end{definition}

Let us note that if $\{\varphi^1_j\}$ and $\{\varphi^2_k\}$ are mutually unbiased orthonormal bases, then $\{\alpha_j\varphi^1_j\}$ and $\{\beta_k\varphi^2_k\}$
are also mutually unbiased  for all complex numbers $\alpha_j,\beta_k\in\T:=\{z\in\C\mid |z|=1\}$.
Mutual unbiasedness is thus actually a property of the corresponding projections.
Namely, define $\Po^1(j)=\kb{\varphi^1_j}{\varphi^1_j}$ and $\Po^2(k)=\kb{\varphi^2_k}{\varphi^2_k}$.
Then \eqref{eqn:mu-bases} translates into
\begin{align}\label{eqn:mu-pvms}
\tr{\Po^1(j)\Po^2(k)}=1/d \, .
\end{align}
We will say that a family of PVMs whose elements satisfy \eqref{eqn:mu-pvms} are mutually unbiased.
Actually, from this we see that mutual unbiasedness admits a direct generalisation to all $d$-outcome POVMs.

\begin{definition}
Let $J$ be an index set and $\{\Eo^j \mid j\in J\}$ a set of POVMs, with each $\Eo^j$ being based on a space $I_j$, where $|I_j| = d$. We say that the POVMs $\Eo^j$, $j\in J$, are \emph{mutually unbiased (MU)} if
\begin{align*}
\tr{\Eo^j(i_j)\Eo^k(i_k)}=1/d
\end{align*}
for all $j,k\in J$ with $j\neq k$, and $i_j \in I_j$, $i_k \in I_k$.
\end{definition}

It is known that there are at most $d+1$ 
mutually unbiased bases (MUBs) in $\hi_d$ \cite{Wootters1989} and that
sets of three MUBs can be constructed in any dimension \cite{KlappRot2004}.
If $d$ is a prime power, that is, $d=p^n$ for some prime number $p$ and $n\geq 1$, 
then $d+1$ MUBs can be constructed in $\mathcal{H}_d$ \cite{Wootters1989}. However, for values of $d$ 
that are not prime powers, this limit does not seem to be reached. For example, there
is extensive numerical evidence suggesting that in $d=6$ no set of  
three MUBs may be extended to a set with more MUBs \cite{Butterley2007}.

\subsection{SICs and their generalisation}
	
\begin{definition}\label{def:SIC}
A POVM $\Go$, with $d^2$ outcomes, acting on $\hi_d$ is \emph{symmetric informationally complete} (SIC POVM or simply SIC) if
\begin{align}
& \tr{\Go(k)\Go(\ell)} = \frac{1}{d^2(1+d)} \qquad \forall k\neq\ell , \label{eqn:SIC1} \\
& \tr{\Go(k)^2} = \frac{1}{d^2} \qquad \forall k  \, . \label{eqn:SIC2}
\end{align}
\end{definition}

It follows from this definition that
\begin{align}
\tr{\Go(k)} = \sum_\ell \tr{\Go(k)\Go(\ell)}=\frac{1}{d} \quad \forall  k  \, .\label{eqn:SIC3} 
\end{align}
This, together with \eqref{eqn:SIC2}, implies that for each $k$, $d\Go(k)$ is a rank-1 projection.

A SIC POVM $\Go$ is automatically informationally complete, meaning that no two states give the same expectations for all $\Go(k)$.
This property follows from the fact that the operators $\Go(k)$ form a basis of the vector space $\mathcal{L}(\hi_d)$ of linear operators on $\hi_d$, which can be shown
directly from the definition or obtained as a consequence of the following interesting lemma \cite{ColCorbDurtGross2005}.

\begin{lemma}
Let $\Go$ be a SIC on $\hi_d$. 
Then the $d^2$ operators 
\begin{equation}
d\Go(k)- t I\, , \quad \textrm{where} \quad t=\frac{1}{d}\left(1-\sqrt{\frac 1{d+1}}\right) \, , 
\end{equation}
 form an orthogonal basis with respect to the Hilbert-Schmidt inner product in $\mathcal{L}(\hi_d)$. 
\end{lemma}

\begin{proof}
Consider  $k\neq\ell$, then
\begin{align*}
\tr{(d\Go(k)- t I)(d\Go(\ell)- t I)}=dt^2-2t+\frac 1{d+1}=0
\end{align*}
since it is easily verified that the above value of $t$ is a root of this polynomial.
\end{proof}

The existence of a SIC POVM has been proved analytically for a number of low-dimensional cases \cite{ScottGrassl2010,Appleby2012a,Appleby2012b} and numerically verified (as of 2010) up to the dimension $67$ \cite{ScottGrassl2010}.

Since the existence of SIC POVMs in an arbitrary dimension is an open problem, we find it useful to introduce the following generalisation. We emphasize, however, that unlike the previous definitions this one has only mathematical motivation.

\begin{definition} 
A collection of selfadjoint operators $\{\Go(i)\}_{i=0}^{d^2-1}$ on $\hi_d$ that satisfiy the SIC conditions \eqref{eqn:SIC1}-\eqref{eqn:SIC2}, and the normalisation $\sum_i \Go(i)=\id$, is called a {\em SIC system}.
\end{definition}

Note that positivity of the $\Go(i)$'s is not required in the definition of a SIC system.

\subsection{The SIC-MUB connection in $\hi_2$}\label{sec:qubit}

We show how a complete set of MUBs can be extracted from a SIC POVM in the case of a qubit system. 
A SIC $\Go$ consists of four positive rank-1 operators 
\begin{equation}
	\Go(k)=\frac{1}{4}(I + \boldsymbol{s}_k\cdot \boldsymbol{\sigma}),\label{eq:SICform}
\end{equation}
where $k\in \mathbb{Z}_4$, $\no{\boldsymbol{s}_k} =1$, and $\boldsymbol{\sigma}= (\sigma_x,\sigma_y,\sigma_z)$ 
is the vector whose elements are the three Pauli matrices. 
The SIC condition reads 
\begin{equation*}
\tr{\Go(i) \Go(j)} = \frac{1}{12} = \frac{1}{8} (1+ \boldsymbol{s}_i\cdot \boldsymbol{s}_j)
\end{equation*}
that is,
$$
\boldsymbol{s}_i \cdot \boldsymbol{s}_j = -\frac{1}{3}\quad \text{if}\quad i\neq j.
$$
This means that the vectors $\boldsymbol{s}_i$ point to the vertices of a regular tetrahedron that lies on the
surface of the Bloch sphere. In particular, their sum is the null vector, which ensures that  $\sum_k \Go(k)=\id$. 

One can form three pairs of binary marginal POVMs:
\begin{equation}\label{eqn:sic-mub-d2}
\begin{split}
\Eo^{1}(+)=\Go(0)+\Go(1),&\quad \Eo^{1}(-)=\Go(2)+\Go(3),\\
\Eo^{2}(+)=\Go(0)+\Go(2),&\quad \Eo^{2}(-)=\Go(1)+\Go(3),\\
\Eo^{3}(+)=\Go(0)+\Go(3),&\quad \Eo^{3}(-)=\Go(1)+\Go(2).
\end{split}
\end{equation}
The positive operators $\Eo^{k}(\pm)$ are of the form
\begin{equation}
	\Eo^{k}(\pm) = \frac{1}{2}(I \pm \boldsymbol{m}_k \cdot \boldsymbol{\sigma}).
\end{equation}
where
\begin{equation}
	\boldsymbol{m}_k := \frac{1}{2}(\boldsymbol{s}_0 + \boldsymbol{s}_k),\quad k=1,2,3.
\end{equation}
It is easily seen that
\begin{equation}
	\boldsymbol{m}_k \cdot \boldsymbol{m}_\ell =\frac13\delta_{k\ell}.
\end{equation}

In order to diagonalize the $\Eo^{k}(\pm)$, we introduce the unit vectors 
$\boldsymbol{n}_k$ via $\boldsymbol{m}_k = \frac{1}{\sqrt{3}}\boldsymbol{n}_k\equiv\lambda\boldsymbol{n}_k$, 
and PVMs $\Po^{k}(\pm) := \frac{1}{2}(I \pm \boldsymbol{n}_k\cdot\boldsymbol{\sigma})$.  
Then	
\begin{align}
	\Eo^{k}(\pm) &= \lambda\Po^k(\pm) + \frac{1}{2}(1-\lambda)I\label{eqn:smearing}\\
	&= \frac{1}{2} (1+\lambda)\Po^{k}(\pm) + \frac{1}{2}(1-\lambda)\Po^{k}(\mp).\label{eqn:spectral}
\end{align}
The  POVMs $\sfe^{k}$ are thus found to arise as \emph{smearings} 
of the PVMs $\sfp^{k}$. 
Furthermore, one has
\begin{align}
\tr{\Po^{k}(\pm) \Po^{\ell}(\pm)}=\frac{1}{2}\quad \mathrm{for} \; k\neq \ell.		
\end{align}
Hence, the three PVMs $\Po^{1}$, $\Po^{2}$ and $\Po^{3}$ are mutually unbiased.
	 
In order to construct a SIC  from a  complete set of MUBs, one can reverse the steps taken 
in the above construction. Given a system of three MUBs, $\{\fii_+^{k},\fii_-^{k}\}$, $k=1,2,3$, define the PVMs $\kb{\fii_\pm^{k}}{\fii_\pm^{k}}=:\Po^{k}(\pm)\equiv\frac12(I\pm\boldsymbol{n}_k \cdot \boldsymbol{\sigma})$. 
Then apply the smearing according to \eqref{eqn:smearing} for some $\lambda\in[0,1]$ to define the three POVMs
$\Eo^{k}$.  
We note that these POVMs are mutually unbiased independently of the value of $\lambda$, i.e.,
\begin{align}
\tr{\Eo^{k}(\pm) \Eo^{\ell}(\pm)} = \frac 12\quad\text{for}\quad k\ne \ell \, .
\end{align}

Next, observe that the system of equations \eqref{eqn:sic-mub-d2}
can be solved; by summing all those triples of the positive operators $\Eo^{k}(\pm)$ that share one of 
the four operators $\Go(i)$ and repeatedly using the fact that $\Go(0)+\Go(1)+\Go(2)+\Go(3)=I$ one finds:
\begin{equation}
\begin{split}
\Eo^{1}(+)+\Eo^{2}(+)+\Eo^{3}(+)&=2\Go(0)+I,\\
\Eo^{1}(+)+\Eo^{2}(-)+\Eo^{3}(-)&=2\Go(1)+I,\\
\Eo^{1}(-)+\Eo^{2}(+)+\Eo^{3}(-)&=2\Go(2)+I,\\
\Eo^{1}(-)+\Eo^{2}(-)+\Eo^{3}(+)&=2\Go(3)+I.
\end{split}
\end{equation}
Thus, given the POVMs $\Eo^{k}$, one can define operators $\Go(i)$ as follows:
\begin{equation}\label{eqn:mub-sic-d2}
\begin{split}
\Go(0)&:=\frac12\left(\Eo^{1}(+)+\Eo^{2}(+)+\Eo^{3}(+)-I\right),\\
\Go(1)&:=\frac12\left(\Eo^{1}(+)+\Eo^{2}(-)+\Eo^{3}(-)-I\right),\\
\Go(2)&:=\frac12\left(\Eo^{1}(-)+\Eo^{2}(+)+\Eo^{3}(-)-I\right),\\
\Go(3)&:=\frac12\left(\Eo^{1}(-)+\Eo^{2}(-)+\Eo^{3}(+)-I\right).
\end{split}
\end{equation}
Equations \eqref{eqn:mub-sic-d2} entail immediately that  the operators $\Go(k)$ are selfadjoint and $\sum_k\Go(k)=I$.
But at this point, we need to choose $\lambda={1}/{\surd{3}}$ to make the operators $\Go(k)$ positive.
With that choice, it is straightforward to verify that $\Go$ is actually a SIC.

\section{From SICs to MUBs}\label{sec:SIC->MUB}

\subsection{Marginals of a SIC POVM}\label{sec:margPOVM}
	
We now consider the route of constructing (up to) $d+1$ MUBs from the $d^2$ 
elements of a SIC $\sfg$ on $\hi_d$. First, note that
any partition of $\sfg$ into disjoint bins constitutes a marginal POVM, where each of its positive
operators is given as the sum of the SIC elements contained in one of the bins.
We will focus on partitions into $d$ bins of size $d$; these have the property that the positive
operators obtained by summing the SIC elements of each bin are unit trace.
\begin{definition}\label{def:d-part}
A partition of a set of $d^2$ elements, $\{a_1,a_2,\dots,a_{d^2} \}$, into $d$ bins of $d$ elements will be called a {\em $d$-partition}.
\end{definition}
\noindent The structure of a $d$-partition will be required as it is the only kind of partition such 
that the positive operators obtained as the sums of the SIC elements within each bin have trace equal to 1.
As we deal only with $d$-partitions we will also occasionally refer to them simply as {\em partitions}.

Clearly there exist many $d$-partitions. If we place the $d^2$ elements of $\{a_1,a_2,\dots,a_{d^2} \}$ 
into a $d \times d$-array,  then 
letting the rows be bins gives a first $d$-partition; and choosing the columns as bins 
gives a second (see Fig.~\ref{fig:1-overlap}). These two $d$-partitions, $\cp^{(1)}$ and $\cp^{(2)}$, which we 
shall  refer to as the {\em Cartesian partitions}, have the following relevant 
property: any pair of bins, one from each of the two partitions, share exactly one element of 
the $d\times d$ array. 
\begin{definition}
Any two distinct $d$-partitions, $\cp^{(k)}=\{\cpk_1\dots,\cpk_d\}$ and 
$\cp^{(\ell)}=\{\cpl_1,\dots,\cpl_d\}$, satisfying 
the condition that the cardinalities of all bin intersections $|\cpk_\mu\cap\cpl_\nu| =1$ will be said to 
have (or share) the {\em 1-overlap property}.
\end{definition}
For a given bin $\cpk_\nu$ we define the trace-one operator 
\begin{align}
\Eo^{k}(\nu) = \sum_{\Go(i) \in \cpk_\nu} \Go(i). \label{eq:margeffects}
\end{align}
For each $d$-partition $\cpk$ of $\sfg$, the set of positive operators $\{\Eo^{k}(\nu)\}_{\nu=1}^d$
forms a marginal POVM $\sfe^{k}$ of the given SIC.

\begin{figure}
\begin{center}
\includegraphics[scale=0.2]{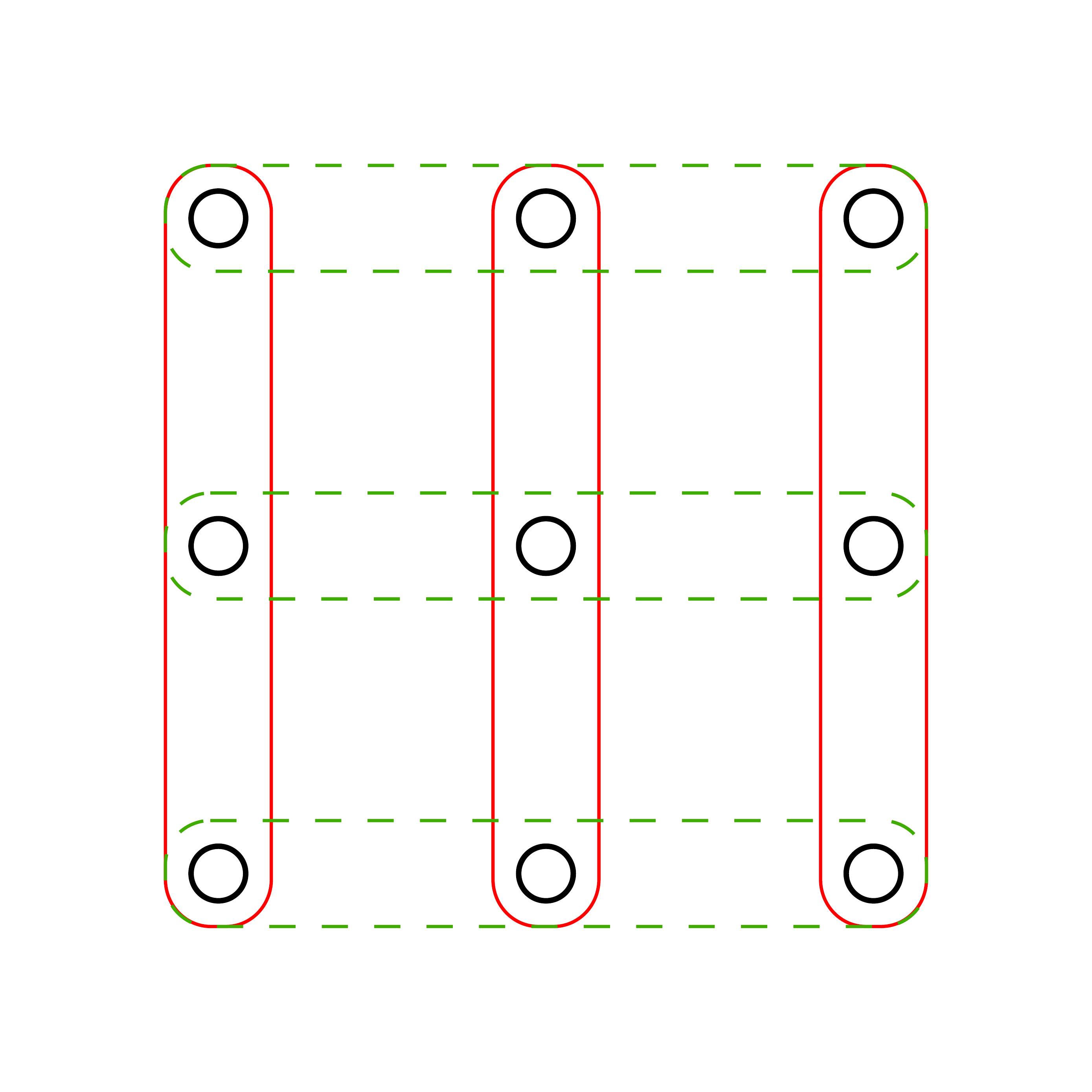}
\end{center}
\caption{Cartesian partitions. (Color online.)}
\label{fig:1-overlap}
\end{figure}

The following is an immediate  consequence of the defining properties of a SIC.

\begin{theorem}\label{thm:mup-marg}
Let $\sfg$ be a SIC, and let $\cpk,\,\cpl$ be two $d$-partitions of $\sfg$ with the 1-overlap property, 
with associated marginal POVMs $\sfe^{k}$ and $\sfe^{\ell}$.
Then
\begin{align}
\tr{\Eo^{k}(\mu)}&=1,\quad \mu=1,\dots,d;\label{eqn:muptr1}\\
\tr{(\Eo^{k}(\mu))^2}&=\frac2{d+1},\quad \mu=1,\dots,d;\label{eqn:siccompatd-1}\\
\tr{\Eo^{k}(\mu) \Eo^{k}(\nu)}&=\frac 1{d+1},\quad\mu\ne\nu;\label{eqn:siccompatd-2}\\
\tr{\Eo^{k}(\mu) \Eo^{\ell}(\nu)}&=\frac 1d,\quad \mu,\nu=1,\dots,d\quad \text{for }k\ne\ell.\label{eqn:mupdef}
\end{align}
\end{theorem}
\begin{proof}
The trace-1 property is immediate.
\begin{widetext}
\begin{align*}
\tr{\left(\Eo^{k}(\nu)\right)^2} &= \sum_{\Go(i) \in \cpk_\nu} \tr{(\Go(i))^2}+ \sum_{\substack{\Go(i),\Go(j)\in\cpk_\nu\\ i\neq j}} \tr{\Go(i) \Go(j)} 
					=  \frac{2}{d+1};\\
\tr{\Eo^{k}(\nu) \Eo^{k}(\mu)} &= \sum_{\Go(i) \in \cpk_\nu} \sum_{\Go(j) \in \cpk_\mu} \tr{\Go(i) \Go(j)}  
					= \frac{d^2}{d^2(d+1)} = \frac{1}{d+1}\quad \mathrm{where} \ \mu\neq \nu;\\
\tr{\Eo^{k}(\nu) \Eo^{\ell}(\mu)}=& \tr{\Go\big(i(\nu,\mu)\big)^2} +
					 \sum_{\substack{\Go(n) \in \cpk_\nu\\ n\ne i(\nu,\mu)}}
					 \sum_{\substack{\Go(m) \in \cpl_\mu\\m\ne i(\nu,\mu)}} \tr{\Go(n) \Go(m)} \notag \\
					 &+\sum_{\substack{\Go(m) \in \cpl_\mu\\m\ne i(\nu,\mu)}} \tr{\Go\big(i(\nu,\mu)\big) \Go(m) }
					 + \sum_{\substack{\Go(n) \in \cpk_\nu\\ n\ne i(\nu,\mu)}} \tr{\Go\big(i(\nu,\mu)\big)\Go(n)} \notag\\
					=& \frac{1 +d - 1}{d^2} = \frac{1}{d},\quad\text{where}\ \nu\ne\mu.
\end{align*}
\end{widetext}
Here $\Go\big(i(\nu,\mu)\big)$ is the SIC element for which $i(\nu,\mu)$ is the unique shared label in 
$\cpk_\nu \cap \cpl_\mu$.  
\end{proof}

\begin{remark}\label{rem:sicmupprops}
We note for later use that (\ref{eqn:mupdef}) implies (\ref{eqn:muptr1}) and (\ref{eqn:siccompatd-2}) 
implies (\ref{eqn:siccompatd-1}) due to the normalisation of the POVMs.
\end{remark}

It is worth noting that the effects $\Eo^{k}(\mu)$ in the marginal POVMs $\sfe^{k}$ are linearly independent.

\begin{proposition}\label{li}
Let $\sfe^{k}$ be a POVM satisfying conditions $(\ref{eqn:muptr1})$-$(\ref{eqn:siccompatd-2})$. The effects $\Eo^{k}(1),\ldots,\Eo^{k}(d)$ are linearly independent.
\end{proposition}

\begin{proof}
Let $Y_\nu:=\Eo^{k}(\nu)-\frac{1}{d+1}I$, $\nu=1,2,\dots,n$. By equations (\ref{eqn:muptr1}), (\ref{eqn:siccompatd-1}) and (\ref{eqn:siccompatd-2}), 
\begin{equation*}
\tr{Y_\nu \Eo^{k}(\mu)}=
\begin{cases}
\frac{1}{d+1} & \mathrm{if} \quad\nu=\mu,\\
0 & \mathrm{if} \quad\nu\neq\mu.
\end{cases}
\end{equation*}
\noindent
Suppose that $\sum_{\mu=1}^n c_\mu \Eo^{k}(\mu)=0$. Then, 
\begin{align*}
0&={\rm tr}\Big[Y_\nu\sum_{\mu=1}^n c_\mu \Eo^{k}(\mu)\Big]=\sum_{\mu=1}^n c_\mu \tr{Y_\nu \Eo^{k}(\mu)}\\
&=c_\nu\frac{1}{d+1},\quad \nu=1,\dots,n,
\end{align*}
which implies $c_\nu=0$ for all $\nu=1,\dots,n$. Hence the linear independence follows. 
\end{proof}

\begin{definition}
We will say that a collection of $d$-outcome POVMs is {\em SIC-compatible} if they are marginals of a common SIC. 
\end{definition}

A collection of SIC-compatible POVMs is {\em jointly measurable} or {\em compatible}, in the probabilistic sense that 
the SIC constitutes a joint observable. This scenario constitutes a remarkably strong form of joint measurability:
a set of $n$ random variables with (say) $d$ outcomes  gives rise to $d^n$ joint event combinations, so that
normally one would have to construct a joint observable with $d^n$ outcomes. In the present case one can
have up to $d+1$ marginals but the SIC constitutes a joint observable for all of them with only $d^2$ outcomes.

Theorem \ref{thm:mup-marg} 
thus states  significant necessary conditions for the SIC-compatibility of a collection of mutually unbiased POVMs. 
A point of significance is that the 1-overlap property postulated for the system of $d$-partitions 
is necessary for the resulting marginal POVMs to be mutually unbiased. 
This will be proven below.

\subsection{More on the 1-overlap property}\label{sec:o-op}

\begin{lemma}\label{lem:numer-of-partitions}
The number of $d$-partitions of $d^2$ elements sharing the 1-overlap property is at most $d+1$ and  at least 3.
\end{lemma}
	
\begin{proof}		
Suppose there are $n$ $d$-partitions of a label set $\{a_1,a_2,\dots,a_{d^2}\}$. 
Consider all bins of the $n$ partitions that contain element $a_1$, say, and 
denote them by $\cpk_1$, where $1 \leq k \leq n$; so $\cpk_1 \cap \cpl_1 = \{a_1 \}$ 
for all $k\neq \ell$. Hence all other elements 
within these bins cannot be repeated. There are $d^2-1$  other elements  in $\{a_2,a_3,\dots,a_{d^2}\}$ 
remaining to be distributed over the $\cpk_1$, and each of these bins 
must contain another $d-1$ elements, so the number of partitions $n \leq (d^2-1)/(d-1) = d+1$.

Next, if the elements of $\{a_1,a_2,\dots,a_{d^2}\}$ (or their labels, for simplicity) are arranged in a $d\times d$ array,
			\begin{align*}
				\begin{matrix}
					1 & \dots & d\\
					\vdots & \ddots & \vdots \\
					(d-1)d+1 & \dots & d^2
				\end{matrix}
			\end{align*}
then, as noted after Definition \ref{def:d-part}, the Cartesian partitions of rows and columns,
$\cp^{(1)}$ and $\cp^{(2)}$, are $d$-partitions that share the 1-overlap property.
By choosing these Cartesian partitions first, we are now restricted in the construction of any further partition: 
in order to satisfy the 1-overlap property with respect to the Cartesian partitions, no bin, besides 
those in $\cp^{(1)}$ and $\cp^{(2)}$, can contain two elements from the same row or 
column within this array.\\
A third $d$-partition, $\cp^{(3)}$,  that shares the 1-overlap property with the two Cartesian 
partitions can be obtained as follows. Relabel the elements of the $d\times d$-array in matrix form, $(a_{ij})$. 
Let $\cp^{(3)}_1$ be the bin consisting of the main left-to-right
diagonal, $\cp^{(3)}_1=\{a_{11},a_{22},\dots,a_{dd}\}$. Now consider the $2(d-1)$ subsets of the
array that are given by as many diagonals parallel to this main diagonal. Denote by $S_{1j}$ the diagonal subsets 
starting with the
elements $a_{1j}$, $j\ge 1$, and by $S_{d+2-i,1}$ those starting with  $a_{i1}$, $i\ge 2$. 
Then the disjoint bins $\cp^{(3)}_1=S_{11}$
and $\cp^{(3)}_\nu=S_{1\nu}\cup S_{\nu 1}$, $\nu\ge 2$, form a $d$-partition that shares 
the 1-overlap property with $\cp^{(1)}$ and $\cp^{(2)}$.
\end{proof}

From now 
on, any partition said to share the 1-overlap property is immediately assumed to share it with respect to the 
Cartesian partitions.
			
			We recall the definition of an important structure within combinatorics. 
A \emph{Latin square of order $d$} is a $d \times d$ array $A$ in which each point on 
$A$ contains an element from a symbol set $S$, where $\abs{S}=d$, such that each row and column within $A$ 
contains every element from $S$. In other words, no column or row within $A$ has repetition of any element from $S$. 
			
			Consider two Latin squares of order $d$ with symbol sets $S$ and $T$, and construct a $d\times d$ 
array of ordered pairs $(s,t)$ placed at point $(i,j)$ (where $i$ and $j$ are the respective row and 
column indices), where $s\in S$ is the element found at that point in the array in the first Latin square, 
and likewise $t\in T$ is the element found at that point in the array in the second Latin square. These two 
Latin squares are \emph{mutually orthogonal} if every ordered pair $(s,t)$ is unique, that is no two points in the 
array have the same ordered pair.
			
With these definitions at hand we obtain the following proposition. We note that the second part concerning the connection
with orthogonal Latin squares was noted and proven by Mikl\'os Hartmann \cite{Hartmann2012}.

\begin{proposition}\label{prop:L-sq}
There is a one-to-one correspondence between the set of $d$-partitions $\cpk$ of a 
$d\times d$ array that share the 1-overlap property with both Cartesian partitions  
$\cp^{(1)},\,\cp^{(2)}$ and
the set of Latin squares of order $d$. Furthermore, two $d$-partitions 
with the 1-overlap property with respect to each other and the Cartesian partitions correspond to mutually 
orthogonal Latin squares.
\end{proposition}
\begin{proof}
Consider a single partition $\cpk$ of a $d \times d$ array  
sharing the $1$-overlap property with the Cartesian partitions. 
Each bin $\cpk_\nu$ can be viewed as a path through this array,
and due to the disjoint nature of the bins no two paths intersect at any point. 
Since each point of the array 
belongs to a bin, it  is incident with a unique path. Thus one can attach to each point 
the label of the path it belongs to; in other words, we have a mapping $(i,j)\mapsto \nu^{(k)}_{i,j}$ from  point labels $(i,j)$
to path labels $\nu^{(k)}_{i,j}\in \{1,\dots,d \}$.  In order to share the 1-overlap property with the Cartesian partitions, 
no bin/path  $\cpk_\nu$ can contain more than one element from any column or row within this array, and this 
corresponds to having no paths  which can move strictly vertically or 
horizontally at any point, nor can they return to a row or column they have already been incident on. This means 
that no value of $\nu^{(k)}_{i,j}$ may be repeated within any row or column 
in the $d\times d$ array $\{\nu^{(k)}_{i,j}\}$. In other words we have proved that the array  $\{\nu^{(k)}_{i,j}\}$ is a Latin square.

Conversely any Latin square of order $d$ induces a path labelling on the elements of the given $d\times d$ array,
and it is immediate that the resulting paths form a partition that shares the 1-overlap property with the Cartesian 
partitions.
				
Consider now two partitions $\cpk, \cpl$ that share the 1-overlap 
property between each other and the Cartesian partitions.  Since they share the 1-overlap 
property with the 
Cartesian partitions they both correspond to Latin squares of order $d$, which we denote by 
$\{\nu^{(k)}_{i,j}\}$ and $\{\nu^{(\ell)}_{i,j}\}$ respectively. If we consider the $d\times d$ array 
$\{(\nu^{(k)}_{i,j},\nu^{(\ell)}_{i,j})\}$, the 1-overlap property implies that the ordered pairs $(\nu^{(k)}_{i,j},
\nu^{(\ell)}_{i,j})$ are distinct for all $i,j$; i.e., the Latin squares are orthogonal. Indeed, if it were that $(\nu^{(k)}
_{i,j},\nu^{(\ell)}_{i,j})=(\nu^{(k)}_{n,m},\nu^{(\ell)}_{n,m})$, then we would have two paths, $\cpk_{\nu^{(k)}_{i,j}}$ 
and $\cpl_{\nu^{(\ell)}_{i,j}}$, from the two different partitions, $\cpk$ and $\cpl$ 
respectively, which intersect at both
$(i,j)$ and $(n,m)$, which contradicts the 1-overlap property.
We can therefore conclude that two partitions 
$\cpk$ and $\cpl$ that share the 1-overlap property with 
each other and the Cartesian partitions correspond to mutually orthogonal Latin 
squares.
			\end{proof}

\subsection{SIC-compatibility for mutually unbiased POVMs}\label{Subsec:SIC-MUPs}

We now turn to the problem of determining conditions under which a given set of mutually unbiased POVMs can be marginals of a common SIC.	
Recall that any marginal POVM $\sfe^{k}$ is specified by a partition $\cpk$ of $\sfg$ such that each positive operator 
$\Eo^{k}(\nu)$ is the sum of the SIC elements from one of the bins:
\[
\Eo^{k}(\nu) = \sum_{\Go(i) \in \cpk_\nu} \Go(i).  
\]

\begin{proposition}\label{prop:1-ov}
	If a collection of at least three mutually unbiased POVMs $\sfe^{k}$ are given as marginals of a SIC then the partition associated with each 
	$\sfe^{k}$ is a $d$-partition $\cpk$. Moreover, $d$-partitions associated with different marginals 
	$\sfe^{k}$ and $\sfe^{\ell}$ satisfy the 1-overlap property.
\end{proposition}
			
\begin{proof}
	We begin with a collection of mutually unbiased POVMs $\sfe^{k}$ which do not have a fixed number of elements, and each 
$\Eo^{k}(\nu)$ is composed of an unspecified number of elements from the SIC. We show first that in 
order to satisfy the mutual unbiasedness condition ${\rm tr}\bigl[{\Eo^{k}(\nu) \Eo^{\ell}(\mu) }\bigr]= 1/d$,  
each of the $\sfe^{k}_\nu$ must be composed of $d$ SIC elements; from this we deduce the 1-overlap property.
				
	Consider two partitions $\cpk$ and $\cpl$ containing $m_k$ and $m_\ell$ 
bins, respectively. We do not make any restrictions on the number of elements from the SIC that are in any of the bins 
in $\cpk$ and $\cpl$, and so we denote $\bigl|{\cpk_\nu}\bigr|= n^{(k)}_\nu$, and likewise 
$\bigl|{\cpl_\mu}\bigr|= n^{(\ell)}_\mu$, where $\sum_\nu n^{(k)}_\nu = \sum_\mu n^{(\ell)}_\mu = d^2$. We denote 
$\bigl|{\cpk_\nu \cap \cpl_\mu}\bigr|=a^{k,\ell}_{\nu,\mu}$, that is $a^{k,\ell}_{\nu,\mu}$ is the number 
of elements that overlap between $\cpk_\nu$ and $\cpl_\mu$. These numbers must satisfy 
$\sum_\nu a^{k,\ell}_{\nu,\mu} = n^{(\ell)}_\mu$ and $\sum_\mu a^{k,\ell}_{\nu,\mu} = n^{(k)}_\nu$ because every element 
in $\cpl_\mu$ must appear in exactly one of the bins in $\cpk$, and likewise for 
$\cpk_\nu$. With this notation we obtain
				
\begin{equation}\begin{split}\label{eq:a-ov}
	\frac{1}{d} &= \tr{\Eo^{k}(\nu) \Eo^{\ell}(\mu)} \\
	&= \frac{1}{d^2}a^{k,\ell}_{\nu,\mu} 
				+ \frac{1}{d^2(d+1)}(n^{(k)}_\nu n^{(\ell)}_\mu - a^{k,\ell}_{\nu,\mu})\\
				&= \frac{1}{d^2(d+1)} \left(d a^{k,\ell}_{\nu,\mu}+n^{(k)}_\nu n^{(\ell)}_\mu \right).
\end{split}
\end{equation}
Summing over $\mu$ we obtain		
\begin{align*}
	\sum_\mu \frac{1}{d} = \frac{m_\ell}{d} = \frac{1}{d^2(d+1)} \left( d n^{(k)}_\nu + d^2 n^{(k)}_\nu \right) = \frac{1}{d} n^{(k)}_\nu,
\end{align*} 
hence, $m_\ell=n^{(k)}_\nu$, and likewise by summing over $\nu$ we find that $m_k = n^{(\ell)}_\mu$. 
This means that all bins of $\cpk$ have the same number of elements, $m_\ell$, the number of bins in 
$\cpl$, and likewise all bins in $\cpl$ have the same number of elements, $m_k$,
the number of bins in $\cpk$. 
We may thus drop the $\nu$ and $\mu$ subscripts from $n^{(k)}_\nu$ and $n^{(\ell)}_\mu$, 
respectively, and simply consider $n^{(k)}$ and $n^{(\ell)}$. With this new notation we have 
$n^{(k)} m_k = n^{(\ell)} m_\ell =m_k m_\ell= d^2$. 

We now consider three marginals:
${\rm tr}\bigl[{\Eo^{k}(\nu) \Eo^{j}(\sigma)} \bigr] = {\rm tr}\bigl[{\Eo^{\ell}(\mu) \Eo^{j}(\sigma)}\bigr] =1/d$ for all $\nu,\mu,\sigma$. 
By repeating the preceding argument we find that $n^{(k)} =n^{(\ell)} = m_j$, and so the bins in partitions 
$\cpk$ and $\cpl$ contain the same numbers of elements; this also shows us that 
$m_k=m_\ell$, and so $d^2 = m_k m_\ell = m_k^2$, hence $m_k = n^{(k)} = d$, and likewise for 
the other partitions.
				
We can now return to equation \eqref{eq:a-ov}:
\begin{align*}
	\frac{1}{d}=\frac{1}{d^2(d+1)}(d a^{k,\ell}_{\nu,\mu} + &n^{(k)} n^{(\ell)}) = \frac{1}{d^2(d+1)}(d a^{k,\ell}_{\nu,\mu} + d^2). 					
\end{align*}
We conclude that $a^{k,\ell}_{\nu,\mu} = 1$.
Since  the choice of bins was arbitrary, this result holds in general.
\end{proof}

\begin{corollary}\label{cor:3+mups}
For every SIC in $\hi_d$ there exist sets of at least three and at most $d+1$ marginals that are mutually unbiased POVMs. 
For $d=6$ there are no more than three marginal POVMs that are mutually unbiased.
\end{corollary}

\noindent The first statement is a direct consequence of Lemma \ref{lem:numer-of-partitions} and 
Theorem \ref{thm:mup-marg}. The latter statement follows from the fact that there are no pairs of mutually 
orthogonal Latin squares of order 6 \cite{Stinson1984}, so that the 1-overlap property required by Proposition \ref{prop:1-ov}
cannot be satisfied for more than three partitions.

\subsection{Commutative case}\label{sec:commmups}
\label{subsec:comm}

We are interested in constructing three or more  mutually unbiased POVMs $\sfe^{k}$ that arise as marginals from a SIC $\sfg$. We have seen so far
that it is necessary to define such marginals via $d$-partitions of $\sfg$ that satisfy the 1-overlap property. 
Thus the mutually unbiased POVMs $\sfe^{k}$ will satisfy conditions 
\eqref{eqn:muptr1}-\eqref{eqn:mupdef}. If one could show that such mutually unbiased POVMs can be 
constructed so that each POVM
$\sfe^{k}$ is commutative, that is the elements $\Eo^{k}(\nu)$ are mutually 
commuting for fixed $k$, then there is a possibility to extract a system of MUBs from these mutually unbiased 
POVMs.
 
In this commutative case, the elements $\Eo^{k}(\nu)$ of each $\sfe^{k}$ share an orthonormal eigenbasis and so have
a decomposition in terms of the projections $\Po^{k}(i)$ onto the elements of the eigenbasis:
\begin{align}
	\Eo^{k}(\nu) =\sum_{i=1}^d \lambda^{k}_{\nu,i} \Po^k(i)\label{eq:E-k-nuspec},
\end{align}
where the $\lambda^{k}_{\nu,i}$ are the eigenvalues of $\Eo^{k}(\nu)$ satisfying $\Eo^{k}(\nu) \Po^{k}(i) = \lambda^{k}_{\nu,i} \Po^{k}(i)$. Equation \eqref{eq:E-k-nuspec} shows that the commutative POVM 
$\sfe^{k}$ is a smearing of the PVM $\Po^{k}:=\{\Po^{k}(i)\}_{i=1}^d$. 

 By making use of equation \eqref{eq:E-k-nuspec}, and the orthogonality of the $\Po^{k}(i)$, we can find 
some requirements the eigenvalues $\big(\lambda^{k}_{\nu,i}\big)$ must satisfy.
			
Firstly, we consider the trace-one property \eqref{eqn:muptr1} and the normalisation  $\sum_\nu \Eo^{k}(\nu) = I$. 

\begin{subequations}
\begin{align}
	1&= \tr{\Eo^{k}(\nu)} = \sum_i \lambda^{k}_{\nu,i} {\rm tr}\bigl[{\Po^{k}(i)}\bigr] \nonumber\\
	&= \sum_i \lambda^{k}_{\nu,i} = \boldsymbol{\lambda}^{k}_\nu \cdot \mathbb{I};\label{eq:doubstoch1}\\
	I&= \sum_\nu \Eo^{k}(\nu) = \sum_\nu  \sum_i\lambda^{k}_{\nu,i} \Po^{k}(i) 
	= \sum_i \sum_\nu \lambda^{k}_{\nu,i}  \Po^{k}(i),\notag \\
	&\text {so } \sum_\nu \lambda^{k}_{\nu,i} = 1,\quad \mathrm{or, equivalently} 
	\quad \sum_\nu \boldsymbol{\lambda}^{k}_\nu = \mathbb{I}. \label{eq:doubstoch2}
\end{align}
\end{subequations}
where $\boldsymbol{\lambda}^{k}_\nu=(\lambda^{k}_{\nu,1},\dots,\lambda^{k}_{\nu,d})$ is the vector of eigenvalues of $\Eo^k(\nu)$, $\mathbb{I}=(1,\dots,1)$ in equation \eqref{eq:doubstoch1}, and in equation \eqref{eq:doubstoch2} 
we have made use of the linear independence of the $\Po^{k}(i)$. If we form a matrix, 
$\Lambda^{k}=\bigl(\lambda^{k}_{\nu,i} \bigr)$, of eigenvalues for the elements of the $k$-th mutually unbiased POVM, 
then equations \eqref{eq:doubstoch1}, \eqref{eq:doubstoch2} are seen to entail that $\Lambda^{k}$ 
is a doubly stochastic matrix.
			
If we examine equations \eqref{eqn:siccompatd-1} and \eqref{eqn:siccompatd-2} using equation \eqref{eq:E-k-nuspec} we find that
\begin{subequations}
\begin{align}
	\frac{2}{d+1}&= \tr{\left(\Eo^{k}(\nu)\right)^2} = \sum_i \left(\lambda^{k}_{\nu,i}\right)^2 
	= \abs{\boldsymbol{\lambda}^{k}_\nu}^2 \label{eq:dotprod1},\quad\ \ \, \\
			\frac{1}{d+1}&= \tr{\Eo^{k}(\nu) \Eo^{k}(\mu)} = \sum_i \lambda^{k}_{\nu,i}\lambda^{k}_{\mu,i} 
			= \boldsymbol{\lambda}^{k}_\nu \cdot \boldsymbol{\lambda}^{k}_\mu\label{eq:dotprod2}.
\end{align}
\end{subequations}
Note that, in line with Remark \ref{rem:sicmupprops}, the unit trace property \eqref{eq:doubstoch1} 
is an immediate consequence of the (self-)overlap
conditions \eqref{eq:dotprod1}, \eqref{eq:dotprod2} together with the normalisation condition $\sum_\nu \Eo^{k}(\nu) = I$.
Alternatively, the self-overlap property \eqref{eq:dotprod1} follows from the mutual overlap property \eqref{eq:dotprod2} 
and the unit trace and normalisation conditions.

Making use of equations \eqref{eq:doubstoch1} and \eqref{eq:dotprod1} one can see  that the endpoints 
of the eigenvalue 
vectors $\boldsymbol{\lambda}^{k}_\nu$ lie on the intersection between the $d-1$-dimensional 
hyperplane $\boldsymbol{\lambda}^{k}_\nu \cdot \mathbb{I} =1$
and the $d-1$-sphere centred at the origin with radius $[2/(d+1)]^{1/2}$. Since the $\boldsymbol{\lambda}^{k}_\nu$
have equal lengths and overlaps, their endpoints span the vertices of a regular 
simplex with centre $\frac 1d\mathbb{I}$ in that hyperplane. This can be confirmed by considering the vectors
\begin{align}
\boldsymbol{r}_\nu = \boldsymbol{\lambda}^{k}_\nu - \tfrac1d\mathbb{I},
\end{align}
These vectors connect the centre of the simplex with its vertices and they have constant lengths and overlaps:
\begin{subequations}
	\begin{align}
		\abs{\boldsymbol{r}_\nu}^2 & = \abs{\boldsymbol{\lambda}^{k}_\nu}^2 
		+ \abs{(1/d)\mathbb{I}}^2 - (2/d) \boldsymbol{\lambda}^{k}_\nu \cdot \mathbb{I} \notag\\
		&= \frac{2}{d+1} - \frac{1}{d} = \frac{d-1}{d(d+1)},\\
		\boldsymbol{r}_\nu \cdot \boldsymbol{r}_\mu &= \boldsymbol{\lambda}^{k}_\nu \cdot \boldsymbol{\lambda}^{k}_\mu 
		+ \abs{(1/d)\mathbb{I}}^2 - (1/d)(\boldsymbol{\lambda}^{k}_\nu + \boldsymbol{\lambda}^{k}_\mu)\cdot \mathbb{I}\notag\\
		&= \frac{1}{d+1} - \frac{1}{d} = -\frac{1}{d(d+1)}.\label{eq:rdot}
\end{align}
\end{subequations}
Equation \eqref{eq:rdot} can be rewritten as 
\begin{align*}
	\boldsymbol{r}_\nu \cdot \boldsymbol{r}_\mu = \abs{\boldsymbol{r}_\nu}\abs{\boldsymbol{r}_\mu} \cos\theta 
	&= \frac{d-1}{d(d+1)}\cos\theta = -\frac{1}{d(d+1)}\\
	\Rightarrow \cos \theta &= -\frac{1}{d-1},
\end{align*}
where the constancy of the angle $\theta$ between $\boldsymbol{r}_\nu$ and $\boldsymbol{r}_\mu$ reflects the
regularity of the simplex. We summarize:

\begin{proposition}\label{prop:simplex}
      Let $\sfe^{k}$ be a commutative POVM whose elements 
      $\Eo^{k}(\nu) =\sum_i \lambda^{k}_{\nu,i} \Po^k(i)$ have unit trace (hence the eigenvalue
      matrix $\Lambda^{k}$ is doubly stochastic) and 
      satisfy the overlap condition \eqref{eq:dotprod2}.
      Then the $d$ eigenvalue vectors $\boldsymbol{\lambda}^{k}_\nu\in\mathbb{R}^d$ point  to the vertices of a regular $d-1$-simplex 
      lying on the intersection of the $d-1$-dimensional hyperplane $\boldsymbol{\lambda}^{k}_\nu \cdot \mathbb{I}=1$  
       and the $d-1$-sphere of radius $(2/(d+1))^{1/2}$ centred at the origin.
\end{proposition}
			
Proposition \ref{prop:simplex} provides a geometric interpretation of properties that a set of $d$ positive operators 
with the same eigenbasis must satisfy in order to arise as a  commutative marginal of a SIC associated with a 
$d$-partition. These properties are also used in proving the following theorem. 
			
\begin{theorem}\label{thm:commMUPs<->MUBs}
      Let $\sfe^{k}$ and $\sfe^{\ell}$ be two commutative POVMs whose elements have unit trace and satisfy the overlap property
      \eqref{eq:dotprod2}. If their decompositions in terms of eigenprojections are given as in equation \eqref{eq:E-k-nuspec}, i.e.
    \begin{equation}\label{eq:kl-spec}
	\Eo^{k}(\nu) = \sum_i \lambda^{k}_{\nu,i}\Po^{k}(i), \quad \Eo^{\ell}(\mu) = \sum_j \lambda^{\ell}_{\mu,j}\Po^{\ell}(j), 
    \end{equation}
      then the following equivalence holds:
   \begin{align*}
	&\tr{\Eo^{k}(\nu) \Eo^{\ell}(\mu)} = \frac{1}{d}\; \forall\; \mu,\nu\notag \\
	&\qquad \iff  \tr{\Po^{k}(i)\Po^{\ell}(j)}=\frac{1}{d}\; \forall \; i,j\,.
   \end{align*}
\end{theorem}
			
\begin{proof}
	Given equation \eqref{eq:kl-spec}, we verify that the condition 
${\rm tr}\bigl[{\Po^{k}(i)\Po^{\ell}(j)}\bigr] = 1/d$ for all $i,j$ leads directly to 
${\rm tr}\bigl[{\Eo^{k}(\nu) \Eo^{\ell}(\mu)}\bigr]=1/d$ for all $\nu,\mu$. Since
\begin{align}
	\tr{\Eo^{n}(\mu)} = \sum_{i=1}^d \lambda^{n}_{\mu,i}\, \tr{\Po^{n}(i)} 
	= \sum_{i=1}^d \lambda^{n}_{\mu,i} =1,\quad n=k,\ell,
\end{align}
we have that 
\begin{align}
	\tr{\Eo^{k}(\mu) \Eo^{\ell}(\nu)}
       &= \sum_{i=1}^d \lambda^{k}_{\mu,i}\sum_{j=1}^d \lambda^{\ell}_{\nu,j}\, \tr{\Po^{k}(i )\Po^{\ell}(j)}\notag \\
	&=\left(\sum_{i=1}^d \lambda^{k}_{\mu,i}\right)\left(\sum_{j=1}^d \lambda^{\ell}_{\nu,j} \right)\frac{1}{d} = \frac{1}{d},
\end{align}
whenever $k \neq \ell$, and so the smeared operators are mutually unbiased. (Note that this implication does not rely on the 
overlap property \eqref{eq:dotprod2}.)

Next, we prove the converse implication.
By making use of equation \eqref{eq:kl-spec}, we see that
\begin{align}
	\tr{\Eo^{k}(\nu) \Eo^{\ell}(\mu)} = \frac{1}{d}
	&=\sum_{i,j} \lambda^{k}_{\nu,i}\, \lambda^{\ell}_{\mu,j}\, {\rm tr}\left[{\Po^{k}(i) \Po^{\ell}(j)}\right] \notag \\
	&=: \sum_{i,j} \lambda^{k}_{\nu,i}\, q^{k \ell}_{ij} \, \lambda^{\ell}_{\mu,j} ,
\end{align}
or equivalently, 
\begin{align}
	\Lambda^{k} Q^{k\ell} (\Lambda^{\ell})^{T} = \frac{1}{d} U,\label{eq:MUPrewrite}
	\end{align}
where $Q^{k\ell} = \bigl(q^{k\ell}_{ij}\bigr) = \bigl({\rm tr}\bigl[{\Po^{k}(i) \Po^{\ell}(j)}\bigr]\bigr)$ and $U_{ij} = 1$ $\forall$ $i,j$. 
 Because the rows of $\Lambda^{k}$ (and likewise for $\Lambda^{\ell}$) correspond to the vertices of a regular simplex, they 
 are linearly independent, and this means that $\Lambda^{k}$ (and $\Lambda^{\ell}$ respectively) is invertible. 
 Therefore, there exist two matrices, $\Gamma^{k}$ and $\Gamma^{\ell}$, such that
\[
\Gamma^{k}\Lambda^{k}=\Gamma^{\ell}\Lambda^{\ell}=I,
\]
\noindent
and so
\[
Q^{k\ell}=\frac{1}{d}\,\Gamma^{k}U(\Gamma^{\ell})^T.
\]  
Note that, thanks to the (row) stochasticity of  $\Lambda^{k}$ and $\Lambda^{\ell}$, the matrices $\Gamma^{k}$ and 
$(\Gamma^{\ell})^T$ are such that the sum of each row in the first case and the sum of each column in the second case 
is one. Indeed, a matrix $A$ is such that  the sum of each row is one if and only if $A\,\mathbb{I}^T=\mathbb{I}^T$ 
where, $\mathbb{I}=(1,1,\dots,1)$. Thanks to the row stochasticity of  $\Lambda^{k}$, we have 
$\mathbb{I}^T=I\,\mathbb{I}^T=(\Gamma^{k}\Lambda^{k})\,\mathbb{I}^T=\Gamma^{k}(\Lambda^{k}\,\mathbb{I}^T)
=\Gamma^{k}\,\mathbb{I}^T$ which proves that the sum of each row in $\Gamma^{k}$ is $1$. The same reasoning 
shows that the sum of each row in $\Gamma^{\ell}$ is $1$ which implies that the sum of each column in $(\Gamma^{\ell})^T$ 
is $1$. We note that, in general, $\Gamma^{k}$ and $\Gamma^{\ell}$ are not stochastic since they are not positive definite. 
It follows that 
\[
\Gamma^{k}U=U;\quad\quad U(\Gamma^{\ell})^T=U
\]
and
\[
Q^{k\ell}=\frac{1}{d}\,\Gamma^{k}U(\Gamma^{\ell})^T=\frac{1}{d}\,U,
\] 
which completes the proof.
\end{proof}
We summarize the implications of the results from this section, using the fact that the conditions of 
Theorem \ref{thm:commMUPs<->MUBs} are satisfied for mutually unbiased POVMs that are marginals of a SIC.

\begin{corollary}
Let $\sfg$ be a SIC which possesses a family of commutative marginals $\sfe^{k}$ generated by $d$-partitions
that share the 1-overlap property (hence they are mutually unbiased). 
Then the orthonormal eigenbases shared by elements of the same marginal are mutually unbiased.
\end{corollary}
						
It is an open question whether every SIC possesses commutative marginals arising from a $d$-partition.

\section{From MUBs to SICs}\label{sec:MUB->SIC}

We consider next the problem of constructing a SIC from a given collection of MUBs. As a preliminary 
investigation we establish necessary conditions for a family of mutually unbiased POVMs to arise as marginals of a
SIC. In the case where there are $d+1$ mutually unbiased POVMs satisfying certain properties to be
specified below, it is possible to establish a formula for the corresponding  SIC elements. 
This requires us to invert the system of equations \eqref{eq:margeffects},
\[
\Eo^{k}(\nu) = \sum_{\Go(i) \in \cpk_\nu} \Go(i)\quad (\nu=1,\dots, d,\ k=1,\dots,d+1),
\]
bearing in mind that the underlying $d$-partitions satisfy the 1-overlap property. We can take guidance from the construction
shown in the qubit case. We begin with some important combinatorial observations.

\subsection{Combinatorial interlude}

We consider an $(d+1)\times d$ array of points with a system of paths with the property that every path
is of length $d+1$ containing one element from each row and any two paths share the {\em 1-overlap property},
that is, they intersect in exactly one point.

\begin{lemma}\label{lem:d^2paths}
	Let $\mathcal{E}$ be a $(d+1)\times d$ array of points. 
	There are at most $d^2$ paths $p_i$ through $\mathcal{E}$ containing exactly one point from each row
	such that any two paths intersect with each other exactly once.
\end{lemma}
\begin{proof}
			Consider all paths that start, and therefore intersect, on the $j^{th}$ element in 
the first row, i.e., at $\mathcal{E}_{1j}$; from this point on no two of these paths 
may intersect at any further point. Since no point within the second row can appear within 
two paths beginning at $\mathcal{E}_{1j}$, we can have at most $d$ possible paths passing through 
the second row.  Given that we have $d$ elements within the first row, we can therefore 
have at most $d^2$ paths with the required properties.
\end{proof}
	
There is a fundamental connection between the array $\mathcal{E}$ with its path system described above
and the $d\times d$ array used to define $d$-partitions of a SIC with the 1-overlap property.

\begin{proposition}\label{prop:ABequiv}
			A $d\times d$ array $A$ with a complete set of $d+1$ partitions satisfying the 1-overlap property 
			between bins, each of order $d$, of different $d$-partitions is equivalent to a $(d+1)\times d$ 
			array $B$ with a complete set of $d^2$ downwards paths of length $d+1$ 
			that satisfy the 1-overlap property between any pair of paths; the equivalence consists in 
			the correspondence of the points and paths in  $A$ with the paths and points in $B$, respectively. 
\end{proposition}
		
\begin{proof}
	We will start with the $d\times d$ array $A=\{a_{i,j}\}$, which has partitions satisfying the 1-overlap property 
between bins/paths $\mathcal{P}^{(k)}_\nu$, $\mathcal{P}^{(l)}_\mu$, of different partitions $\mathcal{P}^{(k)}$, $\mathcal{P}^{(l)}$ of $A$. As noted in Section \ref{sec:o-op}, 
$\abs{\mathcal{P}^{(k)}_\nu \cap\mathcal{P}^{(\ell)}_\mu} = 1$ when $k \neq \ell$ and, by Lemma \ref{lem:numer-of-partitions}, there 
exist at most $d+1$ partitions. We assume that $A$ possesses $d+1$ partitions with the 1-overlap property. 
We can construct a $(d+1)\times d$ array $B$ as follows: each point of $B$ 
corresponds to a bin/path, and the $d$ bins from a given partition are arranged in rows. For example, $B$ could be the array $\{b_{k,\nu}\}=\{\mathcal{P}_\nu^{(k)}\}$.

Now, for each point $a_{i,j}$ in $A$ 
we can define a path $p_{i,j}$ on $B$ as follows: $p_{i,j}=\{\mathcal{P}^{(1)}_{r_1(i,j)},\dots,$ $\mathcal{P}^{(d+1)}_{r_{d+1}(i,j)} \}$ where, $a_{i,j}\in\mathcal{P}^{(k)}_{r_k(i,j)}$, $k=1,\dots,d+1$. The paths on $B$ are strictly downwards and this corresponds to the fact that each point in $A$ appears only once in a given partition.

Note that each path $p_{i,j}$ on $B$ contains $d+1$ bins $\mathcal{P}^{(l)}_{r_{l(i,j)}}$, $l=1,\dots,d+1$, each bin $\mathcal{P}^{(l)}_{r_l(i,j)}$ contains the point  $a_{i,j}$ plus $d-1$ points of $A$, and, by the 1-overlap property between the bins $\mathcal{P}^{(l)}_{r_{l(i,j)}}$, the point $a_{i,j}$ is the only one they have in common. Therefore a total of $(d+1)(d-1)+1=d^2$ distinct points (all the points in $A$) are involved in each path on $B$. This means that for any point $a_{i,j}$ and for any path $p_{\nu,\mu}=\{\mathcal{P}^{(1)}_{r_1(\nu,\mu)},\dots,\mathcal{P}^{(d+1)}_{r_{d+1}(\nu,\mu)} \}$, there is an index $k$ such that $a_{i,j}\in\mathcal{P}^{(k)}_{r_{k(\nu,\mu)}}$. Therefore, two arbitrary paths $p_{\nu,\mu}$, $p_{i,j}$ on $B$ must overlap at least once. Indeed, $a_{i,j}\in\mathcal{P}^{(l)}_{r_{l(i,j)}}$, for all $l$, and, by the previous reasoning, there is a $k$ such that $a_{i,j}\in\mathcal{P}^{(k)}_{r_{k(\nu,\mu)}}$. That  implies $\mathcal{P}^{(k)}_{r_{k(i,j)}}\cap\mathcal{P}^{(k)}_{r_{k(\nu,\mu)}}\neq\emptyset$ and then $\mathcal{P}^{(k)}_{r_{k(\nu,\mu)}}=\mathcal{P}^{(k)}_{r_{k(i,j)}}$ since $\mathcal{P}^{(k)}_{r_{k(\nu,\mu)}}$ and $\mathcal{P}^{(k)}_{r_{k(i,j)}}$ belong to the same partition of $A$ and can have a common point only if they coincides.

 The 1-overlap property defined on $A$ means also that the two paths $p_{\nu,\mu}$, $p_{i,j}$ overlap exactly once:  if they overlap more than once then, there are two indices, $k,l$, such that $\mathcal{P}^{(k)}_{r_k(\nu,\mu)}=\mathcal{P}^{(k)}_{r_k(i,j)}$ and $\mathcal{P}^{(l)}_{r_l(\nu,\mu)}=\mathcal{P}^{(l)}_{r_l(i,j)}$ and then, there exist two bins, $\mathcal{P}^{(k)}_{r_k(\nu,\mu)}$, $\mathcal{P}^{(l)}_{r_l(\nu,\mu)}$,  from different partitions of $A$ such that $a_{\nu,\mu},a_{i,j}\in \mathcal{P}^{(k)}_{r_k(\nu,\mu)}$ and $a_{\nu,\mu},a_{i,j}\in \mathcal{P}^{(l)}_{r_l(\nu,\mu)}$. This 
violates the 1-overlap property defined on $A$. Thus we have proved that the 1-overlap property between bins in different partitions of $A$ implies the 1-overlap property between any pair of paths on  $B$.
			
Conversely, let a $(d+1)\times d$ array $B$ be given with the path system as described. 
The 1-overlap property defined on $B$ 
entails that, by Lemma \ref{lem:d^2paths}, we have at most $d^2$ strictly downward paths 
where $d$ paths intersect on any given point. Here we assume that that number is $d^2$. 
We define a $d \times d$ array $A=\{a_{\nu,\mu}\}$ where, 
$a_{\nu,\mu}:=p^{\mu}_\nu$ is the $\mu^{th}$ path on $B$ containing the point $b_{1,\nu}$. With this construction and using the 1-overlap property on $B$, we define the paths on $A$ to correspond to the points in $B$ as follows: to a point $b_{i,j}$, $i\neq 1$, in $B$ there corresponds the path 
on $A$, $\{p^{\mu_1(i,j)}_1,\dots,p^{\mu_d(i,j)}_d\}$, such that  $b_{i,j}\in p^{\mu_k(i,j)}_k$ for each $k=1,\dots,d$, while to a point $b_{1,j}$ there corresponds the path $\{p^1_j,\dots,p^d_j\}$.

Since the paths on $B$ are strictly 
downward, each row of $B$ forms a partition of its paths. Therefore, if we take $d$ paths on $A$, each corresponding to a point in a particular row of $B$, we get a $d$-partition of the points in $A$. That is clearly true for the families of paths corresponding to the points $b_{1,j}$, $j=1,\dots,d$, which define the partition of $A$ in horizontal bins.  It is also true for the family of paths $\{(p^{\mu_1(i,k)}_1,\dots, p^{\mu_d(i,k)}_d)\}$, $k=1,\dots,d$, corresponding to the row $(b_{i,k})$, $i\neq 1$, $k=1,\dots,d$, which cannot overlap. Indeed, if it were 
$p^{\mu_j(i,k)}_j=p^{\mu_j(i,l)}_j$, we would have $b_{i,k}, b_{i,l}\in p^{\mu_j(i,k)}_j$ which means that the path $p^{\mu_j(i,k)}_j$ is not downward. In other words, a set of $d$ 
non-intersecting paths on $A$, which correspond to a row on $B$, form a partition of 
the points on $A$. 

As a consequence of the 1-overlap property on the paths on $B$, any two 
paths on $B$ intersect at only one point, which on $A$ corresponds to the property 
that any two paths from different partitions intersect only once: two paths on $A$ 
intersecting more than once correspond to two points on $B$ which are coincident 
with the same two paths twice, which cannot occur because of the 1-overlap property. For example, if the paths $(p^{\mu_1(i,j)}_1,\dots, p^{\mu_d(i,j)}_d)$, and  $(p^{\mu_1(\nu,\mu)}_1,\dots, p^{\mu_d(\nu,\mu)}_d)$ corresponding to the points $b_{i,j}$ and $b_{\nu,\mu}$, $i,\nu\neq 1$, respectively, overlap twice then, there are two indices $l$, $s$, such that  $p^{\mu_l(i,j)}_l= p^{\mu_l(\nu,\mu)}_l$ and $p^{\mu_s(i,j)}_s= p^{\mu_s(\nu,\mu)}_s$. That implies $b_{i,j},b_{\nu,\mu}\in p^{\mu_l(i,j)}_l$ as well as $b_{i,j},b_{\nu,\mu}\in p^{\mu_s(i,j)}_s$ which contradicts the 1-overlap property of the paths on $B$. The same reasoning applies if one of the two paths is replaced by an horizontal path on $A$. 
So by starting with $B$ and its path system, we have managed to construct $A$ with its own
path system satisfying the required properties.
\end{proof}

\subsection{Construction of SIC elements from mutually unbiased POVMs}
\label{subsec:B}
We assume that a collection of $d+1$  POVMs $\sfe^{k}$ is given, each with $d$ outcomes, and introduce a $(d+1) \times d$ array $\mathcal{E}$ with points 
$\mathcal{E}_{k\nu} = \Eo^{k}(\nu)$: 
\begin{align}\label{eq:mup-array}
\mathcal{E}=
		\begin{matrix}
			\Eo^{1}(1) & \dots & \Eo^{1}(d)\\
			\vdots & \ddots & \vdots \\
			\Eo^{d+1}(1) & \dots & \Eo^{d+1}(d)
		\end{matrix} .
\end{align}

Let us now assume that these  POVMs $\sfe^{k}$ are marginals of a $d^2$-outcome observable
$\sfg=\{\Go(i)\}_{i=0}^{d^2-1}$; moreover, we assume that each marginal is associated with a $d$-partition,
where any two partitions satisfy the 1-overlap property.
This entails that no two 
elements belonging to the same POVM have an operator $\Go(i)$ in common, whilst two elements from 
different POVMs must share just one. In this case, each $\Go(i)$ belongs to one element from each
of the $d+1$ POVMs. On the array $\mathcal{E}$, these relations correspond to paths $\cp_i$, defined as 
subsets of those elements that share $\Go(i)$. From the above constraints on the $d^2$ positive operators 
two properties of these paths are as follows:
\begin{enumerate}
	\item Given that no $\Go(i)$ can occur within two elements of the same POVM $\sfe^{k}$, 
	the paths must be strictly downwards, i.e. they cannot contain any horizontal routes through the array;
	\item Because any two elements from different $\sfe^{k}$ and $\sfe^{\ell}$ share a single operator $\Go(i)$, any two paths 
	through $\mathcal{E}$ must intersect at a single point: if two paths were to intersect twice, then the 
	points where the intersections occurred would correspond to two elements that have two operators, 
	$\Go(i)$ and $\Go(j)$ say, in common, thereby violating the construction requirements.
\end{enumerate}
By virtue of Proposition \ref{prop:ABequiv}, it is ensured that there are $d^2$ downward paths of length $d+1$
in $\mathcal{E}$ such that any two paths share the 1-overlap property. An example of this construction is given in
Fig.~\ref{fig:paths}.

\begin{figure}
\begin{center}
\includegraphics[scale=0.2]{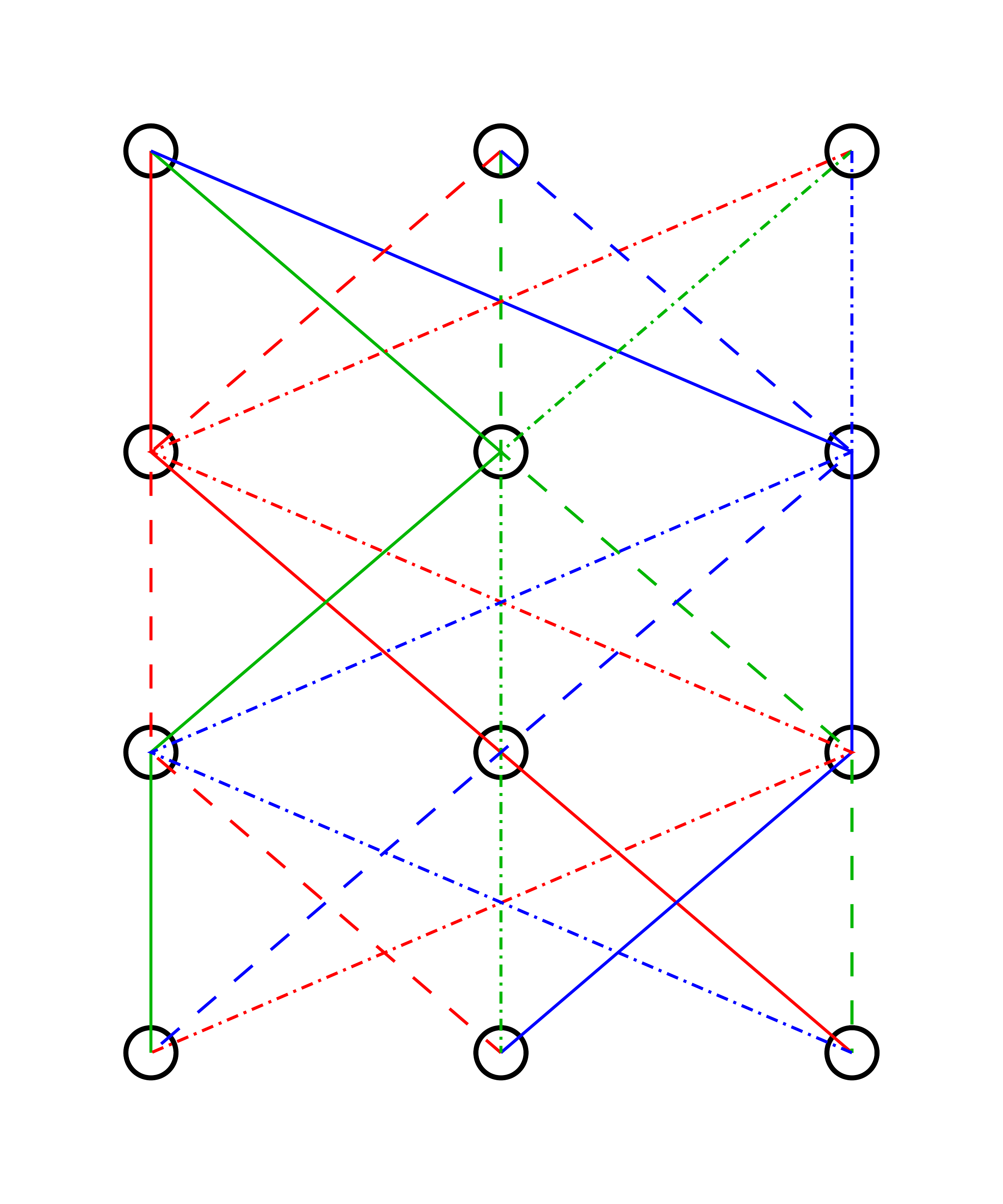}
\end{center}
\caption{$d^2$ Paths with the 1-overlap property for $d=3$.}
\label{fig:paths}
\end{figure}

		We now consider the path $\cp_i$ connecting all points in the array $\mathcal{E}$ 
of marginal POVM elements defined in \eqref{eq:mup-array} which 
contain $\Go(i)$. At this point we define the operators 
\begin{equation}\label{eq:path-ops}
E_i = \sum_k \Eo^{k}\big(\nu(i,k)\big),
\end{equation}
where 
$\Eo^{k}\big(\nu(i,k)\big)$ is the element in the $k$-th POVM $\sfe^{k}$ incident with the path $\cp_i$. 
Note that $E_i$ is the sum of $d+1$ operators, all sharing just $\Go(i)$
 and so all other operators $\Go(j)$  included within the path can occur at most once. 
Each of the $d+1$ operators  $\Eo^{k}\big(\nu(i,k)\big)$ 
is the sum of $d$ SIC elements, one of which is $\Go(i)$, and the remaining $d-1$ elements $\Go(j)$
must be unique to $\Eo^{k}\big(\nu(i,k)\big)$. Hence a total of $(d+1)(d-1)=d^2-1$  elements of $\Go$
other than $\Go(i)$ are required.
 Their sum is equal to $I - \Go(i)$ since $\sum_{j=0}^{d^2-1} \Go(j) = I$. Therefore,
 \begin{align*}
	E_i&= \sum_k \Eo^{k}\big(\nu(i,k)\big) = (d+1) \Go(i) + (I- \Go(i))\\
	& = d \Go(i) + I.
\end{align*}
Hence we obtain the following.

\begin{proposition}\label{prop:magicform}
Assume  a $d^2$-outcome POVM $\sfg$ has $d+1$  marginals $\sfe^{k}$ associated with a complete system of
$d-$partitions satisfying the 1-overlap property. 
Then the elements $\Go(k)$ can be recovered from the $\sfe^k$ via
\begin{align}
	\Go(k) = \frac{1}{d} (E_k - \id)\label{eq:magicform}
\end{align}
where the operators $E_k$ are defined in equation \eqref{eq:path-ops}.
\end{proposition}

We are now ready to prove the main result of this section.

\begin{theorem}\label{thm:magic}
Consider a family of $d+1$ POVMs $\sfe^{k}$ which are mutually unbiased, i.e.,
\[
\tr{\Eo^{k}(\mu) \Eo^{\ell}(\nu)}=\frac1{d},\ k,\ell=1,\dots,d+1,\ \nu,\mu=1,\dots d,
\]
 and assume that the elements also satisfy
\[
\tr{\Eo^{k}(\mu) \Eo^{k}(\nu)}=\frac1{d+1},\ k=1,\dots,d+1,\ \nu\neq\mu.
\]
Further assume that there are $d^2$ sets $\cp_i$, $i=0,1,\dots,d^2-1$, consisting of $d+1$ elements, one chosen from each of  the $\sfe^{k}$,
such that the 1-overlap property is fulfilled. Let operator $E_i$ be the sum of the elements in  $\cp_i$.
Then the $d^2$ operators $\Go(i):=\frac1d(E_i-I)$ form a SIC if and only if they are positive. 
\end{theorem}

\begin{proof} 
We will be using the observation made in 
Remark \ref{rem:sicmupprops}, which entails that under the conditions of the theorem the POVMs also satisfy
\[
\tr{\Eo^{k}(\mu)}=1,\quad \tr{\Eo^{k}(\mu)^2}=\frac2{d+1}.
\]
We verify the SIC conditions \eqref{eqn:SIC1}-\eqref{eqn:SIC2}. 
First,
\begin{align}
	\tr{\Go(i)}= \frac{1}{d}\tr{E_i - I} = \frac{1}{d}(d+1) - 1 = \frac{1}{d}\label{eq:SICpi-i},
\end{align}	
where we have used ${\rm tr}\bigl[{E_i}\bigr]= \sum_k {\rm tr}\bigl[{\Eo^{k}\big(\nu(i,k)\big)}\bigr] = d+1$. 
Secondly,
\begin{align*}
	\tr{\Go(i)^2} &= \frac{1}{d^2} \tr{E_i^2 -2 E_i + I}= \frac{1}{d^2}(\tr{E_i^2} - d - 2),
\end{align*}
Note that $\tr{E_i^2}={\rm tr}\left[{\bigl(\sum_k \Eo^{k}\big(\nu(i,k)\big)\bigr)^2}\right]$ contains $d+1$ 
terms of the form ${\rm tr}\left[{\bigl(\Eo^{k}(\nu)\bigr)^2}\right]=2/(d+1)$ and $d(d+1)$ terms of the form 
$\tr{\Eo^{k}(\mu) \Eo^{\ell}(\nu)}=1/d$ since any two distinct elements in  a given $E_i$ must 
belong to different POVM. 
Therefore,
\[
\tr{(E_i)^2}=2(d+1)/(d+1) + d(d+1)/d = d+3,
\] 
and so
\begin{align}
	\tr{\Go(i)^2} = \frac{1}{d^2}(\tr{E_i^2} - d - 2) = \frac{1}{d^2}.\label{eq:SICpi-ipi-i}
\end{align}
This coincides with the required value for this quantity.

We are ready to determine the trace of $\Go(i) \Go(j)$ for $i\ne j$:
\begin{equation}
\begin{split}
	\tr{\Go(i) \Go(j)}&= \frac{1}{d^2}\tr{E_i E_j - E_i - E_j + I} \\
	&=\frac{1}{d^2}(\tr{E_iE_j} - d -2).\label{eq:trpi-ipi-j}
\end{split}
\end{equation}
The product $E_i E_j$ consists of $(d+1)^2$ terms: $d+1$ will correspond to the product of two elements from 
the same POVM, i.e. $\Eo^{k}\big(\nu(i,k)\big)\Eo^{k}\big(\nu(j,k)\big)$, and $d(d+1)$ will be the product of elements from 
different POVMs, that is, $\Eo^{k}\big(\nu(i,k)\big)\Eo^{\ell}\big(\nu(j,\ell)\big)$. Because the intersection of 
$\cp_i$ and $\cp_j$ contains only one element, it follows that one of the $d+1$ terms 
$\Eo^{k}\big(\nu(i,k)\big)\Eo^{k}\big(\nu(j,k)\big)$ 
will have trace $2/(d+1)$, whilst the remaining $d$ terms will have trace $1/(d+1)$, and the $d(d+1)$ 
terms of the form $\Eo^{k}\big(\nu(i,k)\big)\Eo^{\ell}\big(\nu(j,\ell)\big)$ will have trace $1/d$. Using these values we have that 
\begin{align*}
\tr{E_i E_j}&=\frac{2}{d+1}+\frac{d}{d+1}+\frac{1}{d}(d(d+1))\\
&= \frac{1}{d+1} + d+ 2,
\end{align*} 
and so equation \eqref{eq:trpi-ipi-j} becomes
\begin{align}
	\tr{\Go(i)\Go(j)} = \frac{1}{d^2}(\tr{E_iE_j}-d-2) = \frac{1}{d^2(d+1)}. \label{eq:SICpi-ipi-j}
\end{align}
Finally we verify the normalisation condition:
\begin{align*}
\sum_{i=0}^{d^2-1}\Go(i)=\frac1d\sum_{i=0}^{d^2-1}E_i-dI.
\end{align*}
Each $E_i$ contains $d+1$ POVM elements, each of which is from a different $\sfe^{k}$. Each POVM element
$\Eo^{k}(\mu)$ is shared by exactly $d$ of the operators $E_i$, hence contributes $d$ occurrences to  $\sum_iE_i$.
The sum of all  $d(d+1)$ POVM elements is $(d+1)I$. This gives:
\[
\sum_{i=0}^{d^2-1}\Go(i)=\frac1d d(d+1)I-dI =I.
\]
\end{proof}
This theorem shows that the $\Go$ has all the properties required of a SIC, except, possibly,
the positivity of all $\Go(i)$.
In other words, $\Go$ forms a SIC system and we can rephrase Theorem \ref{thm:magic} as follows.

\begin{corollary}
Under the assumptions of Theorem \ref{thm:magic}, the $d^2$ operators $\Go(i)$ constructed from the mutually unbiased POVMs
form a SIC system.
\end{corollary}

Taking up the construction of commutative POVMs as smearings of mutually unbiased PVMs, we can complete the route from a  family of  $d+1$ MUBs to a SIC system.

Suppose a collection of  MUBs $\{ \varphi_i^{k} \}_{i=1}^d$ ($k=1,\dots,n$) are given,  so that $\Po^{k}(i) = \big|{\varphi_i^{k}}\big\rangle\big\langle{\varphi_i^{k}}\big|$ are mutually unbiased PVMs.
Applying a smearing of the form \eqref{eq:E-k-nuspec} with a stochastic matrix $\Lambda^{k}$ to the PVMs $\Po^{k}$ yields  $n$ POVMs $\sfe^{k}$ as in \eqref{eq:E-k-nuspec}
\[
	\Eo^{k}(\nu) = \sum_{i=1}^d \lambda^{k}_{\nu,i} \Po^{k}(i), \ k=1,\dots,n,\ \nu = 1, \dots, d 
\]

\begin{proposition}\label{prop:traceone}
	POVMs $\Eo^{k}$, $k=1,\dots,n$ obtained as smearings of $n$ mutually unbiased 
	PVMs $\Po^{k}$ via \eqref{eq:MUB-E-k} are mutually 
	unbiased if and only if all elements  $\Eo^{k}(\nu)$ are unit trace operators (that is, the matrices $\Lambda^{k}$
	are doubly stochastic). 	
\end{proposition}

\begin{proof}
The statement that the mutual unbiasedness of the POVMs $\Eo^{k}$ follows from the mutual unbiasedness of the corresponding PVMs
together with unit trace property of the operators $\Eo^{k}(\mu)$ has already been shown in  
Theorem \ref{thm:commMUPs<->MUBs}.

Conversely, assume that the operators $\Eo^{k}(\mu)$ and $\Eo^{\ell}(\nu)$ are mutually unbiased, that is, ${\rm tr}\bigl[{\Eo^{k}(\mu) \Eo^{\ell}(\nu)}\bigr]=\frac1d$ whenever $k\neq \ell$, then summing
over $\nu$ and using the normalisation property $\sum_\nu \Eo^{\ell}(\nu)=I$ immediately gives ${\rm tr}\bigl[ \Eo^{k}(\mu) \bigr]=1$. 
\end{proof}

Combining Theorem \ref{thm:magic} and Proposition \ref{prop:traceone} we have obtained a construction of a SIC system from
a complete family of MUBs.

\section{Comparison with the construction of Wootters}

Wootters \cite{Wootters2006} describes a geometric  structure of dual systems of points and lines that 
turns out to be isomorphic
to the two equivalent combinatorial structures composed of arrays and paths presented here. He begins by 
considering a set of 
$d^2$ points and, upon them, defines \emph{mutually unbiased striations} as follows: a striation is a partitioning 
of the points into $d$ \emph{parallel lines} of $d$ points each, that is, no point is contained in more than one line, 
and two striations are mutually unbiased if any two lines from different striations are coincident at a single point. 

As we have argued before, in this language we are restricted to having at most $d+1$ mutually unbiased striations. 
Assuming that these $d^2$ points are placed into a $d\times d$ array, one begins by defining the horizontal 
and vertical striation in the logical way, and so any further mutually unbiased striation will correspond to a 
Latin square because, as in our construction, we cannot allow for any line to contain two points from the same 
row or column. 
Furthermore, any additional mutually unbiased striations will correspond to Latin squares that are orthogonal
to the first Latin square, and each other, as is the case in our scheme. The difference lies in the association of points 
and lines:  Wootters ascribes to each point $\alpha$ a selfadjoint operator $A_\alpha/d$, and to each 
line $\lambda$ a rank-1 projection $\Po(\lambda)$, satisfying the following properties:
	
	\begin{enumerate}
		\item $\tr{A_\alpha/d}=1/d$;
		\item $\tr{(A_\alpha/d)(A_\beta/d)}=(1/d)\delta_{\alpha\beta}$;
		\item $\sum_{\alpha\in\lambda}A_\alpha/d= \Po(\lambda)$.
	\end{enumerate}
As a consequence of these conditions, a striation corresponds to a basis for a $d$-dimensional Hilbert space -- 
we have $d$ normalized rank-1 projections that are mutually orthogonal -- and two mutually unbiased striations 
correspond to a pair of mutually unbiased bases. 
	
However, as is pointed out in his paper, it is not clear how one constructs the operators $A_\alpha$ so that we 
may find these MUBs without in fact already having a complete set of MUBs to work backwards from. By contrast, 
our method provides us with the necessary operators corresponding to each point -- the elements of a SIC POVM -- 
and by going along a slightly more elongated path -- by constructing striations corresponding to commutative mutually unbiased POVMs 
and then finding the joint eigenbasis for each mutually unbiased POVM -- we are able in some cases to construct a set of MUBs (which 
will not always be complete).
	
Before considering the case for SIC POVMs, we note that the set of $d^2$ points and $d+1$ mutually 
unbiased striations to be considered correspond to an \emph{affine plane of order $d$}. An affine plane of order $d$ is a set of 
$d^2$ points and $d(d+1)$ lines satisfying the following conditions:
	
	\begin{enumerate}
		\item For any two points there exists exactly one line that is coincident with both;
		\item for any point $\alpha$ and line $\lambda$ not containing $\alpha$, there exists a single line that is parallel (i.e. non-intersecting) 
		with $\lambda$ that is coincident with $\alpha$;
		\item there exist three points that are non-collinear.
	\end{enumerate}

It is known that affine planes 
exist if the order is a prime power. The situation for non-prime powers is not very well explored, although it is known that an
affine plane of order 6 does not exist.
	
Following the discussion on MUBs, Wootters proceeds to consider the analogous geometric problem for constructing 
a SIC POVM. Again, this is considered in terms of lines passing through points in such a way that each line corresponds 
to a rank-1 projection $\Po(i) = d \Go(i)$, where $\Go(i)$ is an element of our desired SIC POVM. In order to determine the 
number of points that lie on a given line, and indeed the total number of points to be consider, Wootters adopts 
an idea of  Zauner \cite{Zauner1999}  that provides a connection between the cardinality of a set of points and the 
trace of its corresponding operator. If $M$ is an operator, then let $S_M$ denote the set of points corresponding to it 
within the given geometric construct (so in the previous paragraphs we would have that $S_{A_\alpha} = \{ \alpha\}$, etc.). 
We now look for relations of the form 
	\begin{equation}
		\abs{S_M} = k\; \tr{M} \ \mathrm{and} \ \abs{S_{M_1} \cap S_{M_2}}= k\; \tr{M_1 M_2}
	\end{equation}
for some constant $k$. Since the rank-1 projection $\Po(\lambda)$ associated with the line $\lambda$ 
must be trace one,  it follows that the line contains $k$ points within it. Furthermore, 
$\tr{\Po(\lambda) \Po(\mu)}=1/(d+1)$ entails that distinct  lines $\lambda$ and $\mu$ intersect at $k/(d+1)$ 
points. This overlap must be an integer, hence $k$ must be a multiple of $d+1$. For simplicity, one can choose 
$k=d+1$, and so any two lines intersect once. If one assumes that the operator associated with the entire set is 
the identity, and that each point lies on the same number of lines as any other point, then we have that there 
must exist $k\; \tr{I} = d(d+1)$ points and that each must lie on $d$ lines. We note that the value of $k$ 
was chosen for simplicity, whilst in our work, the value of $k$, i.e. the number of points each line must contain, is 
fixed. 
	
Collecting all the information for the SIC POVM case, we see that one is  considering $d(d+1)$ points and $d^2$ lines 
containing $d+1$ points each, where each point is contained in $d$ lines and any two lines intersect at 
a single point. By performing a relabelling between the points and lines, we see that this construction yields
an affine plane of order $d$. Note that this is similar to the situation for our construction, except that in our case the 
equivalence between the two systems  of points and lines are in fact immediate, rather than the result of choosing a 
particular value of $k$, and that we are again reliant on the dimension being one that allows for an affine plane
to be constructed. 
	
Assuming that an affine plane can be constructed, Wootters  identifies the task of finding a set of selfadjoint 
operators $B_\alpha$ associated 
with each point $\alpha$ such that the line $i$ corresponds to their sum, which is equal to $\Po(i) = d\Go(i)$. In order 
to satisfy the necessary trace conditions, he requires that

	\begin{enumerate}
		\item $\tr{B_\alpha^2}=d/(d+1)^2$;
		\item $\tr{B_\alpha B_\beta}=1/d(d+1)^2$ if $\alpha\neq\beta$ and they share a line;
		\item $\tr{B_\alpha B_\beta}=-1/(d+1)^2$ if $\alpha \neq \beta$ but they don't share a line.
	\end{enumerate}
However, as in the case for MUBs, it is not immediately obvious how these operators are constructed, without 
working backwards (and thereby defeating the purpose). Again, in our work we have the benefit that the operators 
associated with each point -- elements of derived mutually unbiased POVMs -- can be constructed readily.

\section{From covariant SICs to MUBs in odd prime power dimensions}\label{sec:WH-cov}

After this comparison between our combinatorial scheme and Wootters' geometric scheme, we are now free to use geometric language, which is natural here.
In particular, this section is devoted to developing a concrete example of the previous general setting. Given a finite covariant phase space observable acting on a Hilbert space of prime power dimension $d=p^n$ $(p\neq 2)$, we construct $d+1$ marginals that are mutually unbiased POVMs. These marginals turn out to be  smearings of  $d+1$ MUBs, and hence are  commutative (compare with the discussion at the end of \ref{subsec:comm}). The $d+1$ MUBs are then explicitly exhibited.
Rather interestingly, such a construction holds for every finite covariant phase space observable (SIC or otherwise). Finally we  characterize
SIC covariant phase space observables using the results of \ref{subsec:B}.

In this section  $\F$ denotes a finite field with characteristic $p$, different from $2$. Denote by ${\rm Tr} : \F \to \Z_p$ the trace of $\F$ over the cyclic field $\Z_p$, and let $\omega$ be a primitive $p$-root of unit in $\C$. A reader not familiar with finite fields may wish to restrict to the case $\F = \Z_p$  with $p$ an odd prime; in this case the trace ${\rm Tr}$ becomes the identity. This restriction does not simplify any proofs that follow and imposes undue restrictions on the dimensionality of the Hilbert space, hence we shall just consider a generic finite field $\F$.

Consider the  vector space $V=\F\times\F$, endowed with the nondegenerate symplectic form 
\[
\sym{\vv}{\vw}=(v_2w_1-v_1w_2)\,,  
\]
for all $\vv=(v_1,v_2)\,,\, \vw=(w_1,w_2)$ in $V$.
One can think of $V$ as a discrete phase space of a quantum system with finite degrees of freedom \cite{Gibbons2004}, analogous to the usual phase space $\R^2$ of a quantum system with continuous degrees of freedom.

For all $\vv\in V$, the map
$$
\dual{\vv}{\cdot} : V \to \T \qquad \dual{\vv}{\vw} = \omega^{\Tr{\sym{\vv}{\vw}}}
$$
defines a character of the additive abelian group $V$. Moreover, the map $\vv\mapsto \dual{\vv}{\cdot}$ is a homomorphism of the group $V$ into character group $\hat{V}$. Clearly, $\dual{\vv}{\vv} = 1$ and $\dual{\vv}{\vw} = \overline{\dual{\vw}{\vv}}$ for all $\vv,\vw\in V$. The {\em symplectic Fourier transform} of $V$
$$
\ff : \ell^2(V) \to \ell^2(V) \, , \qquad \ff f(\vv) = \frac{1}{|\F|} \sum_{\vu\in V} \dual{\vv}{\vu} f(\vu)
$$
is an invertible linear map with inverse
$$
\ff^{-1} \hat{f}(\vu) = \frac{1}{|\F|} \sum_{\vv\in V} \overline{\dual{\vv}{\vu}} \hat{f}(\vv) \, .
$$

We denote by $\sas$ the set of $1$-dimensional lines of $V$, i.e.,
\begin{equation*}
\sas = \{ L\subset V \mid \F L = L \, , \, \dim_\F L = 1 \} \, .
\end{equation*}
For all $\alpha\in\F$, let
$$
L_\alpha = \F(1,  \alpha)  = \{(u,\alpha u) \mid u\in\F\}
$$
and
$$
L_\infty = \F(0,1) = \{(0,u) \mid u\in\F\} \, .
$$
The following facts are easy to verify.
\begin{enumerate}[(a)]
\item $L_\alpha \cap L_\beta = \{ \vnull \}$ for all $\alpha,\beta\in\F\cup \{\infty\}$ with $\alpha\neq\beta$;
\item $V=\bigcup_{\alpha\in\F\cup \{\infty\}} L_\alpha$;
\item $\sas = \{L_\alpha \mid \alpha\in\F\cup \{\infty\} \}$; 
\item $|\sas| = |\F|+1$.
\end{enumerate}
Hence, each $L\in\sas$ determines an $|\F|$-partition of $V$, i.e.~the set of affine lines $V/L = \{\vv+L \mid \vv\in V\}$ which are parallel to $L$. The $d+1$ partitions $V/L$ associated with the $d+1$ lines $L\in\sas$ share the $1$-overlap property. This is due to the fact that, for $L\neq L'$, the two affine lines $\vv+L$ and $\vu+L'$ meet exactly at one point of $V$.

Let $\hi\coloneqq \ell^2(\F)$ be the $d$-dimensional Hilbert space of complex functions on $\F$.
For each $f\in \hi$ and $\vv=(v_1,v_2)\in V$, define the linear operator $W(\vv)\in\lh$:
\begin{align*}
[W(\vv)f](x)\coloneqq \omega^{\Tr{[v_2(x-2^{-1}v_1)]}} f(x-v_1)\qquad \forall x\in \F\,,
\end{align*}
where $2^{-1}$ is the inverse of $2$ in the field $\F$ and the action is understood to be on the left.
It is an easy calculation to show that
\begin{align*}
 W(\vv+\vw) & =\pair{\vw}{2^{-1}\vv}W(\vv)W(\vw) \qquad \forall \vv,\vw\in V\,.
\end{align*}
Hence, the correspondence
\(
\vv  \mapsto W(\vv)
\)
is a projective representation of $V$ into $\lh$. It is called the  {\em Weyl-Heisenberg representation} of $V$.
The  Stone-von Neumann theorem asserts that the Weyl-Heisenberg representation  is essentially unique (see \cite{Mackey1949}). The following {\em commutation relations} are immediate
\begin{align*}
W(\vv)W(\vw)  & = \dual{\vv}{\vw} W(\vw)W(\vv) \qquad \forall \vv,\vw\in V \, .
\end{align*}

For all $L\in\sas$, we define $|V/L| = |\F|$ operators by formula
\begin{equation}\label{eq:defP}
\Po^L(\vv+L) = \frac{1}{|\F|} \sum_{\vl\in L} \dual{\vv}{\vl} W(\vl) \qquad \forall \vv+L\in V/L \, .
\end{equation}
We then have the following facts.

\begin{proposition}\label{prop:propP}
\begin{enumerate}[{\rm (i)}]
\item For all $\vl\in L$,
\begin{align}
\label{eq:eqW}
W(\vl) = \sum_{\vv+L\in V/L} \overline{\dual{\vv}{\vl}} \Po^L(\vv+L) \, .
\end{align}
\item For all $\vv,\vw\in V$,
\begin{align}
\label{eq:eqW1}
W(\vw) \Po^L(\vv+L) W(\vw)^\ast = \Po^L(\vv+\vw+L) \, .
\end{align}
\item The map $\Po^L : V/L \to \lh$ is a PVM with ${\rm rank}\,{\Po^L(\vv+L)}=1$ for all $\vv\in V$.
\item The PVMs $\{\Po^L\mid  L\in\sas\}$, are mutually unbiased.
\end{enumerate}
\end{proposition}

The proof of the above proposition, though elementary, requires some machinery. In order to keep the exposition short we refer to a forthcoming paper for its proof.

We are considering the vector space $V$ as the phase space for a quantum system.
From this perspective any POVM defined on $V$ can be called a phase space observable.
However, the most important phase space observables are covariant under the Weyl-Heisenberg representation.
We say that a POVM $\Go : V \to \lh$ is a {\em covariant phase space observable} if 
\begin{equation*}
\label{eq:phasespaceobs}
W(\vw)\Go(\vv)W(\vw)^\ast = \Go(\vv+\vw)
\end{equation*}
for all $\vv,\vw\in V$. 
It is well known \cite{Holevo1979} and easy to check in this finite dimensional case that every covariant phase space observable has the form
$$
\Go(\vv) \equiv \Go_T (\vv) := \frac{1}{|\F|} W(\vv)TW(\vv)^\ast
$$
for a unique operator $T\geq 0$ with $\tr{T} = 1$.

For any line $L\in\sas$, we define the \emph{$L$-marginal of $\Go_T$} by 
\begin{equation}\label{eq:defGTL}
\Go_T^L (\vv+L) = \sum_{\vl\in L} \Go_T (\vv+\vl) \qquad \forall \vv+L \in V/L \, .
\end{equation}
It is clear that $\Go_T^L (\vv+L)\geq 0$ and $\sum_{\vv+L \in V/L}\Go_T^L (\vv+L)=\id$, hence $\Go_T^L$ is a POVM on $V/L$.
We further observe that $\Go_T^L$ satisfies the same covariance properties as $\Po^L$, i.e.,
$$
W(\vw)\Go_T^L (\vv+L)W(\vw)^\ast = \Go_T^L (\vv+\vw+L) 
$$
for all $\vv,\vw\in V$. 
In particular, 
\begin{equation}\label{eq:com}
[\Go_T^L (\vv+L) \, , \, W(\vl)] = 0
\end{equation}
for all $\vv\in V$ and $\vl\in L$.

\begin{theorem}\label{prop:propGTL}
Suppose $T\geq 0$ with $\tr{T} = 1$, and let $\Go_T$ be the corresponding covariant phase space observable.
\begin{enumerate}[{\rm (a)}]
\item For each $L\in\sas$, the $L$-marginal $\Go_T^L$ of $\Go_T$ is a smearing of $\Po^L$.
More precisely, for all $\vv\in V$,
\begin{equation}\label{eq:convGTL}
\Go_T^L (\vv+L) = \sum_{\vw+L\in V/L} \Lambda_T^L (\vv-\vw+L) \Po^L (\vw+L) \, ,
\end{equation}
where $\Lambda_T^L : V/L \to [0,1]$ is the probability measure
\begin{equation}\label{eq:defLam}
\Lambda_T^L(\vv+L) = \tr{T\Po^L (-\vv+L)} \, .
\end{equation}
\item The $|\F|+1$ POVMs $\{\Go_T^L \mid L\in\sas\}$ are mutually unbiased.
\end{enumerate}
\end{theorem}

\begin{proof}
(a) By \eqref{eq:com} and the definition of $\Po^L$, we have
$$
[\Go_T^L(\vv+L) \, , \, \Po^L(\vw+L)] = 0 \qquad \forall \vw\in V \, ,
$$
hence, in view of item (iii) of Proposition \ref{prop:propP},
$$
\Go_T^L(\vv+L) = \sum_{\vw+L\in V/L} K(\vv+L,\vw+L) \Po^L(\vw+L) \quad \forall \vv\in V
$$
for some function $K:V/L\times V/L \to \C$. 
Next, we have:
\begin{widetext}
\begin{align*}
K(\vv+L,\vw+L)  = \tr{\Go_T^L(\vv+L) \Po^L(\vw+L)} 
 = \frac{1}{|\F|} \sum_{\vl\in L} \tr{TW(\vv+\vl)^\ast \Po^L(\vw+L) W(\vv+\vl)} .
\end{align*}
Using \eqref{eq:eqW1} this reduces to
\begin{align*}
K(\vv+L,\vw+L)  
=\tr{T\Po^L(\vw-\vv+L)} \equiv \Lambda_T^L (\vv-\vw+L) \, .
\end{align*}

(b) If $L_1,L_2\in\sas$, with $L_1\neq L_2$, then, for all $\vv_1,\vv_2\in V$,
\begin{align*}
 \tr{\Go_T^{L_1}(\vv_1+L_1) \Go_T^{L_2}(\vv_2+L_2)} 
= \frac{1}{|\F|} \sum_{\vw_1+L_1\in V/L_1} \sum_{\vw_2+L_2\in V/L_2} \Lambda_T^{L_1} (\vv_1-\vw_1+L_1) \Lambda_T^{L_2} (\vv_2-\vw_2+L_2)  = \frac{1}{|\F|} \, ,
\end{align*}
\end{widetext}
where the first equality arises from the use of item (iv) of Proposition \ref{prop:propP}, and the final equality is due to $\Lambda_T^{L_i}$ $(i=1,2)$ being probability measures.
\end{proof}

At this point, we recall the qubit example discussed in Section \ref{sec:qubit}.
If we start from three mutually unbiased POVMs $\Po^{k}(\pm) = \frac{1}{2}(I \pm \boldsymbol{n}_k\cdot\boldsymbol{\sigma})$, where the vectors $\boldsymbol{n}_1$, $\boldsymbol{n}_2$ and $\boldsymbol{n}_3$ are orthogonal unit vectors, then it is straightforward to verify that a covariant phase space observable $\Go_T$ is a SIC if the generating operator $T$, now thought of as a state, is rank-1 and satisfies
\begin{align}
\tr{T \boldsymbol{n}_1\cdot\boldsymbol{\sigma}} = \tr{T \boldsymbol{n}_2\cdot\boldsymbol{\sigma}} = \tr{T \boldsymbol{n}_3\cdot\boldsymbol{\sigma}} \, .
\end{align}
This condition simply means that the expectations values of the state $T$ in the measurements of $\Po^{1}$, $\Po^{2}$ and $\Po^{3}$ are equal.
 
We are aiming for a comparable condition in the odd prime power dimensions, but now have to look for more complicated measurements than just $\Po^L$. 
For all $L\in\sas$, we define the PVM $\Qo^L : V/L \to \elle{\hi\otimes\hi}$ by
\begin{equation*}
\Qo^L (\vu+L) = \sum_{\vw\in V/L} \Po^L (\vu+\vw+L) \otimes \Po^L (\vw+L) \, .
\end{equation*}
The measurement outcome distribution of $\Qo^L$ gives the difference of two $\Po^L$-measurements performed on two identical systems.

The following is the main result of the section.

\begin{theorem}
\label{theo:opcharct}
Suppose $T\geq 0$ with $\tr{T} = 1$, and let $\Go_T$ be the corresponding covariant phase space observable.
The following facts are equivalent:
\begin{enumerate}[{\rm (i)}]
\item $\Go_T$ is SIC;
\item $\tr{\Go_T^L (\vv+L) \Go_T^L (\vw+L)} = (1+\delta_{\vv+L,\vw+L}) / (|\F|+1)$ for all $L\in\sas$ and $\vv,\vw\in V$;
\item $\tr{(T\otimes T) \Qo^L (\vv+L)} = (1+\delta_{\vv+L,\vnull+L}) / (|\F|+1)$ for all $L\in\sas$ and $\vv\in V$;
\item ${|\tr{TW(\vv)}|^2} = (1+|\F|\delta_{\vv,\vnull}) / (|\F|+1)$ for all $\vv\in V$.
\end{enumerate}
\end{theorem}
\begin{proof}

(i) $\Leftrightarrow$ (ii)
If $\Go_T$ is SIC, item (ii) follows at once from equations \eqref{eqn:siccompatd-1} and \eqref{eqn:siccompatd-2} in Theorem \ref{thm:mup-marg}. \\
For the converse implication notice that, by item (b) of  Theorem \ref{prop:propGTL}, the $d+1$ marginals  $\Go_T^L$ of a generic phase space observable $\Go_T$ are mutually unbiased.
If $\Go_T$ also satisfies item (ii),  the hypothesis of Theorem \ref{thm:magic} are met. Hence the operators
\[
\frac{1}{d} \left(
\sum_{L\in \sas} \Go_T^L(\vv+L)-I
\right)
\]
form a SIC system that, by Proposition \ref{prop:magicform}, coincides with the phase space observable $\Go_T$. 

(ii) $\Leftrightarrow$ (iii) This is a consequence of
\begin{align*}
&\tr{\Go_T^L (\vv+L) \Go_T^L (\vw+L)}  \\
&\qquad= \sum_{\vu+L\in V/L} \Lambda^L_T (\vu+\vw-\vv+L) \Lambda^L_T (\vu+L) \\
& \qquad= \tr{(T\otimes T) \Qo^L (\vv-\vw+L)} \, .
\end{align*}

(iii) $\Leftrightarrow$ (iv) For all $\vv\in V$, setting $L=\F\vv$, we have
\begin{align*}
& \frac{1}{|\F|} \sum_{\vu\in V} \dual{\vv}{\vu} \Qo^L (\vu+L) = \sum_{\vu+L \in V/L} \dual{\vv}{\vu} \Qo^L (\vu+L) \\
& \;  = \sum_{\vu+L ,\, \vw+L \in V/L} \dual{\vv}{\vu} \Po^L (\vu+\vw+L) \otimes \Po^L (\vw+L)\\
& \; \textrm{(changing $\vu+L \to \vu-\vw+L$)}\\
& \; = \sum_{\vu+L ,\, \vw+L \in V/L} \dual{\vv}{\vu} \overline{\dual{\vv}{\vw}} \Po^L (\vu+L) \otimes \Po^L (\vw+L)\\
& \; \textrm{(by \eqref{eq:eqW})}\\
& \; = W(\vv)\otimes W(\vv)^* \, .
\end{align*}
Therefore,
$$
\frac{1}{|\F|} \sum_{\vu\in V} \dual{\vv}{\vu} \tr{(T\otimes T)\Qo^L (\vu+L)} = |\tr{TW(\vv)}|^2 \, ,
$$
and the equivalence of (iii) and (iv) follows by symplectic Fourier transform for the group $V$.
\end{proof}

A condition equivalent to (iii) in Theorem \ref{theo:opcharct} was already established in \cite{Appleby2007} (see 
equation (17) there).
The part of this condition that we would have anticipated after the qubit example is that the numerical values of the probabilities $\tr{(T\otimes T) \Qo^L (\cdot)}$ are the same for all $L\in\sas$.
It would be interesting to find a physical explanation for the mathematical fact that the correct values are the ones given in (iii) in Theorem \ref{theo:opcharct}.

\section{Conclusion}\label{sec:outlook}
		
		We have obtained an interesting, if incomplete, connection 
between SICs and MUBs that is based on an operational feature, namely joint measurability. 
The results of sections \ref{sec:SIC->MUB} and \ref{sec:MUB->SIC} 
can be summarized as follows: 
\begin{itemize}	
	\item[(a)] For a SIC $\sfg$ we may form $n$ $d$-partitions 
	of $\sfg$ that satisfy the 1-overlap property; the number $n$ is equal to 2 + the 
	maximum number of mutually orthogonal Latin squares of order $d$. The partitions 
	induce $n$ marginal POVMs which are mutually unbiased with respect to each other. 
	
	If the elements of each marginal POVM commute with each other,
	the marginals are smearings of spectral measures whose (rank-1) projections correspond to a system of $n$ MUBs.

	\item[(b)] From a complete set of $d+1$ MUBs one can construct $d+1$ commutative POVMs by smearing the corresponding PVMs, and via suitable
	smearings these POVMs are mutually unbiased. 
	By assuming that each mutually unbiased POVM element is the sum of $d$ operators 
	$\Go(i)$ out of a set $\{\Go(i)\}_{i=0}^{d^2-1}$ such that no two elements from the same POVM share an operator, 
	and any two elements from different POVMs share only one, one can show that the operators 
	$\Go(i)$ form a SIC system, that is they satisfy all properties of a SIC with the possible exception 
	of positivity, which does not follow naturally.
\end{itemize}

With these findings we have established the beginnings of a correspondence
between a SIC and a complete set of $d+1$ MUBs. The bridge providing this connection is
provided by the concept of a complete system of mutually unbiased POVMs that will be marginals of the SIC.
There are, however, obstructions to the existence of this bridge. 

First, the number of mutually unbiased POVMs that
can be obtained as marginals of a SIC is no less than 3 but also no greater than 2 plus the maximal
number of orthogonal Latin squares of order $d$. The maximal number of mutually unbiased POVMs will thus be $d+1$ 
when $d$ is a prime power, but it is unknown whether this number can be obtained for any other
values of $d$. 

Second, the problem remains of extracting MUBs from such a system of mutually unbiased POVMs; we have obtained
a solution only in the case where all the POVMs are commutative.

Third, starting with a system of $d+1$ MUBs (where they exist), one can obtain $d+1$ mutually unbiased POVMs by
applying an appropriate smearing by application of a doubly stochastic matrix to the spectral
measures associated with the MUBs. However, in order to reconstruct a SIC via the formula
\eqref{eq:magicform}, it is required that for the underlying $(d+1)\times d$ array there
exist $d^2$ downward paths of length $d$ which satisfy the 1-overlap property -- this was found again to be
equivalent to the requirement that there are $d-2$ orthogonal Latin squares of order $d$.

To illustrate the significance of these limitations, we revisit the case of a 6-dimensional Hilbert space.  
In this case there do not exist two mutually orthogonal Latin squares of order 6, and so there are no 
systems of mutually unbiased marginals of a SIC with  more than 3 elements.
As a consequence of Proposition \ref{prop:ABequiv}, since there do not exist seven $6$-partitions with the 1-overlap
property in a $6 \times 6$ array, there cannot exist 36 downward paths with the 1-overlap property in a $7 \times 6$ array. 
So even if we hypothesized the existence of 7 MUBs for $d=6$, it is impossible
to construct a SIC from them via formula \eqref{eq:magicform}, although SICs are known to exist for 
$d=6$ \cite{Renes2004}.

To summarize: The existence in $\hi_d$ of a bridge between a complete system of MUBs and 
a SIC via mutually unbiased POVMs and the formula \eqref{eq:magicform}  requires the existence of the 
maximal number $d-1$ of mutually orthogonal Latin squares. 

But even when the conditions of the recovery of a SIC system from a system of $d+1$ mutually unbiased POVMs
are fulfilled, the problem remains of ensuring the positivity of 
the operators $\Go(i)$, $i=1,\dots,d^2$. By construction, these operators 
satisfy $\tr{d\Go(i)}=1$ and $\tr{(d\Go(i))^2}=1$. Using an observation made in \cite{JoLi05}, the operators $d\Go(i)$ are
positive (and in fact rank-1 projections) if, in addition, $\tr{(d\Go(i))^3}=1$.
\\\indent
The positivity of the SIC system elements depend on the doubly stochastic matrices $\Lambda^k$ used in 
smearing our mutually unbiased spectral measures, but these pose their own set of problems. By demanding 
that their coefficients are positive (and hence real) and less than or equal to 1, and that their rows satisfy 
conditions \eqref{eq:dotprod1} and \eqref{eq:dotprod2}, we are able to drastically reduce the number of degrees 
of freedom. Indeed, if we simplify the situation by further requiring that the rows of a given doubly stochastic 
matrix $\Lambda^k$, whose rows satisfy \eqref{eq:dotprod1} and \eqref{eq:dotprod2}, are given as cyclic 
permutations of the elements in the first row, then the entire matrix is determined by a single element of this row, 
with all other elements of the matrix arising as a result. 
However, the value of this variable lies within a continuous interval, and no particular value need immediately stand out as significant. Further to this, it is not known beforehand how many different doubly stochastic matrices need to be used - that is, how many spectral measures are smeared in the same way and how many require a different smearing - in order to retrieve a SIC POVM (in the case of a SIC system, a single matrix - that is, all spectral measures are smeared in the same way - will suffice). The issue of the number of paths does not need discussion here; given that there are $d^{d+1}$ possible downwards paths through a $(d+1)\times d$ array, as described in section \ref{sec:MUB->SIC} (excluding the 1-overlap property here), it is computationally simple to perform a check of all possible paths to see if they form positive operators. 
\\\indent
As an example of these issues, in the case of $d=3$, where we begin with 4 mutually unbiased spectral measures, there are several possible doubly stochastic matrices, of which we choose one, that can be used on 3 of the spectral measures, and then the fourth measure must be smeared by the doubly stochastic matrix whose first row is $(1/2,1/2,0)$ and whose remaining rows are cyclic permutations of this. It is also the case that this doubly stochastic matrix can be used on all four mutually unbiased spectral measures and will still lead to a SIC POVM. 
\\\indent
The SIC--MUB bridge constructed here is based on two combinatorial schemes that were noted to be equivalent;
we showed that this dual combinatorial structure corresponds one-to-one to a dual geometric structure of affine planes
that was posited in a study of Wootters \cite{Wootters2006} in which he laid out analogous schemes for 
constructing MUBs and SICs. The operator systems needed to start Wootters' schemes are not given {\em a priori} 
while in our approach these operators were given and the geometric/combinatorial structures appeared by necessity.
\\\indent
Finally, in Section \ref{sec:WH-cov} we showed that all these facts find a natural application to finite covariant phase space observables in prime power dimension $d=p^n$ $(p\neq 2)$. Indeed, it turns out that a covariant phase space observable on a $d$-dimensional Hilbert space generates $d+1$ mutually unbiased and commutative marginal POVMs, which in turn are smearings of $d+1$ MUBs. Such marginal POVMs are associated to each $d$-partition of the finite phase space into parallel affine lines. Moreover, in Theorem \ref{theo:opcharct}, we  provided various characterizations of SIC covariant phase space observables; in particular, item (iii) establishes an operational link between the SIC property of the covariant phase space observable $\Go_T$ and the $d+1$ MUBs associated with the $d+1$ mutually unbiased marginal POVMs of $\Go_T$.

\section*{Acknowledgements}
T.B. and P.B. gratefully acknowledge support through the White Rose Studentship Network 
{\em Optimising Quantum Processes and Quantum Devices for future Digital Economy Applications}, which provided funds for
a studentship and mutual visits between the project partners. T.H. acknowledges support from the Academy of Finland (grant no.138135). A.~T.~gratefully acknowledges the financial support of the Italian Ministry of Education, University and Research (FIRB project RBFR10COAQ).
					
\bibliographystyle{unsrt}

\end{document}